\pdfoutput=1
\documentclass[12pt,]{article}

\usepackage[utf8]{inputenc} 
\usepackage[T1]{fontenc}    
\usepackage{hyperref}       
\usepackage{url}            
\usepackage{booktabs}       
\usepackage{amsfonts}       
\usepackage{nicefrac}       
\usepackage{microtype}      
\usepackage{lipsum}
\usepackage{tikz}
\usepackage{booktabs}
\usepackage{comment}
\usepackage{filecontents}
\usepackage{amsmath}
\usepackage{calc}
\usepackage{cleveref}
\usetikzlibrary{calc}

\usepackage{microtype}
\usepackage{graphicx}

\graphicspath{{../Figures/}}
\usepackage{booktabs} 
\usepackage{caption}
\usepackage{subcaption}
\usepackage{amsmath}
\usepackage{amsfonts}
\usepackage{authblk}
\usepackage{cite}
\usepackage{algorithm}
\usepackage{algpseudocode}

\usepackage{amsopn,amsthm,bm}
\usepackage{paralist}
\usepackage{multirow}
\usepackage{multicol}
\usepackage{float}
\usepackage{thmtools}
\usepackage{thm-restate}
\usepackage{framed}
\usepackage{hyperref}
\usepackage{enumitem}
\usepackage{tcolorbox}

\usepackage{xcolor}

\usepackage{amsthm}
\newtheorem{proposition}{Proposition}
\newtheorem{assumption}{Assumption}
\newtheorem{theorem}{Theorem}

\newtheorem{lemma}[theorem]{Lemma}
\newtheorem{definition}{Definition}[section]

\usepackage{amsfonts}
\usepackage{amsthm}
\usepackage{amssymb}
\usepackage{epsf}
\usepackage{mathrsfs}
\usepackage{bbm}
\usepackage{colordvi}
\usepackage{amscd}
\usepackage{epsfig}
\usepackage{balance}
\usepackage{mathtools}
\usepackage{wrapfig}
\usepackage{amsfonts}
\usepackage{booktabs}
\usepackage{standalone}
\usepackage{amsfonts}
\usepackage{amsthm}
\usepackage{epsf}
\usepackage{mathrsfs}
\usepackage{bbm}
\usepackage{xcolor}
\usepackage{colordvi}
\usepackage{amscd}
\usepackage{epsfig}
\usepackage{mathtools}
\usepackage{tikz}
\usepackage{lipsum}
\usepackage[normalem]{ulem}

\DeclareMathOperator*{\argmax}{arg\,max}
\DeclareMathOperator*{\argmin}{arg\,min}

\usepackage{tikz}

\definecolor{c1}{HTML}{a1c9f4}
\definecolor{c2}{HTML}{ffb482}
\definecolor{c3}{HTML}{8de5a1}
\definecolor{c4}{HTML}{d0bbff}
\definecolor{cback}{HTML}{FBF6D6}

\usetikzlibrary{positioning,decorations.pathreplacing}

\usepackage{mathtools}
\DeclarePairedDelimiter\floor{\lfloor}{\rfloor}

\title{Efficient Algorithm for Sparse Fourier Transform\\ of Generalized $q$-ary Functions}

\author{Darin Tsui$^{1\dagger}$}
\author{Kunal Talreja$^{1\dagger}$}
\author{Amirali Aghazadeh$^{1, \ast}$}
\affil{$^1$School of Electrical and Computer Engineering,\\ Georgia Institute of Technology\\ $\dagger$Equal contributions\\
*Correspondence to: Amirali Aghazadeh: amiralia@gatech.edu}
\date{\vspace{-5ex}}

\begin{document}
\maketitle

\begin{abstract}
Computing the Fourier transform of a $q$-ary function $f:\mathbb{Z}_{q}^n\rightarrow \mathbb{R}$, which maps $q$-ary sequences to real numbers, is an important problem in mathematics with wide-ranging applications in biology, signal processing, and machine learning. Previous studies have shown that, under the sparsity assumption, the Fourier transform can be computed efficiently using fast and sample-efficient algorithms. However, in most practical settings, the function is defined over a more general space---the space of generalized $q$-ary sequences $\mathbb{Z}_{q_1} \times \mathbb{Z}_{q_2} \times \cdots \times \mathbb{Z}_{q_n}$---where each $\mathbb{Z}_{q_i}$ corresponds to integers modulo $q_i$. Herein, we develop GFast, a coding theoretic algorithm that computes the $S$-sparse Fourier transform of $f$ with a sample complexity of $O(Sn)$, computational complexity of $O(Sn \log N)$, and a failure probability that approaches zero as $N=\prod_{i=1}^n q_i \rightarrow \infty$ with $S = N^\delta$ for some $0 \leq \delta < 1$. We show that a noise-robust version of GFast computes the transform with a sample complexity of $O(Sn^2)$ and computational complexity of $O(Sn^2 \log N)$ under the same high probability guarantees. Additionally, we demonstrate that GFast computes the sparse Fourier transform of generalized $q$-ary functions $8\times$ faster using $16\times$ fewer samples on synthetic experiments, and enables explaining real-world heart disease diagnosis and protein fitness models using up to $13\times$ fewer samples compared to existing Fourier algorithms applied to the most efficient parameterization of the models as $q$-ary functions.
\end{abstract}

\section{Introduction}
Fourier analysis of $q$-ary functions, which map discrete sequences with $q$ alphabets to real numbers, plays a pivotal role in modeling and understanding complex systems across disciplines such as  biology~\cite{wu2016adaptation}, signal processing~\cite{wang2016radar}, and machine learning~\cite{gorji2023walshhadamardregularizer}. Any $q$-ary function $f(\mathbf{m}):\mathbb{Z}_{q}^n\rightarrow \mathbb{R}$ can be expressed in terms of its Fourier transform $F[{\bf k}]$ as,
\begin{equation}
    \label{eq:qary_f}
    f(\mathbf{m}) = \sum_{\mathbf{k} \in \mathbb{Z}_q^n} F[\mathbf{k}]\omega^{\langle \mathbf{m}, \mathbf{k}\rangle} \quad \mathbf{m} \in \mathbb{Z}_q^n,
\end{equation}
where $\omega \coloneqq e^{\frac{2\pi j}{q}}$, and $\mathbf{m}$ and $\mathbf{k}$ are vectors in the $n$-dimensional ring of integers modulo $q$, denoted as $\mathbb{Z}_q^n = \{0,1,\dots,q-1\}^n$. 

Computing the Fourier transform of a $q$-ary function $f$ is a challenging problem. Without any assumptions, this requires obtaining $N=q^n$ samples from $f$ and paying a computational cost of $O(N\log(N))$ to compute the transform using the seminal fast Fourier transform (FFT) algorithm~\cite{cooley1965algorithm}. This algorithm becomes computationally prohibitive for large $q$ or $n$. Fortunately, in practice, most $q$-ary functions have a sparse Fourier transform, meaning $F$ only has a few non-zero coefficients~\cite{aghazadeh2021epistatic,brookes2022sparsity}. The $S$-sparse Fourier transform of a $q$-ary function can be computed using efficient algorithms~\cite{erginbas2023efficiently} with a sample complexity of $O(Sn)$, akin to the compressed sensing theory~\cite{donoho2006cs, baraniukcs2007, candesrobust2006}, and computational complexity of $O(Sn^2\log q)$, significantly faster than sparse regression methods, such as LASSO~\cite{tibshiranilasso1996}. However, despite this progress, in most applications, $f$ is defined over alphabets of different sizes. For instance, consider a model that predicts patient outcomes using $n$ categorical features—such as age group, sex, or chest pain type. This model can only be represented as a function over the space $\mathbb{Z}_{q_1} \times \mathbb{Z}_{q_2} \times \cdots \times \mathbb{Z}_{q_n}$, where $q_i$ reflects the number of categories in feature $i$. Representing this model instead as a $q$-ary function is infeasible, since features often vary in their number of categories (grouping patients by age may require more categories than grouping by sex). As another example, consider a length-$n$ protein sequence, which can be naturally encoded as a $q$-ary sequence in $\mathbb{Z}_{q=20}^n$ for 20 amino acids. When predicting protein fitness from the sequence, it is common to assume only a fraction of the amino acids in each protein site are predictive~\cite{sarkisyan2016local}. In these cases, the function is more effectively modeled using a more generalized notion of $q$-ary function. We call such functions $f:\mathbb{Z}_{q_1} \times \mathbb{Z}_{q_2} \times \cdots \times \mathbb{Z}_{q_n} \rightarrow \mathbb{R}$, which map from the product of spaces with different alphabet sizes to real numbers, \emph{generalized $q$-ary functions}. 


Despite their importance, no efficient algorithm currently exists for computing the Fourier transform of generalized $q$-ary functions. This algorithm would enable several advances, including scalable methods to identify biological interactions in proteins~\cite{brookes2022sparsity,poelwijk2016context}, explain machine learning models~\cite{tsui2024shapzero,tsui2024recovering}, and regularize neural networks in the spectral domain~\cite{aghazadeh2021epistatic,gorji2023walshhadamardregularizer}. Existing workarounds naively require treating the underlying function as $q_\text{max}$-ary with $q_\text{max}=\max_i{q_i}$, which inflates the problem's dimensionality from $\prod_{i=1}^n q_i$ to $q_\text{max}^n$ with no change to the number of non-zero Fourier coefficients $S$. This leads to a substantial increase in sample and computational complexity. For example, modeling protein fitness using a generalized $q$-ary function by focusing on the subset of amino acids predictive of a protein function reduces the problem dimensionality by two orders of magnitude, from $q_\text{max}^{n=12} = 8.9 \times 10^{12}$ to $\prod_{i=1}^{n=12} q_i = 8.6 \times 10^{10}$. This effect is exacerbated as $n$ grows, making the computation of an otherwise feasible Fourier transform inaccessible due to sample and computational overheads. 

Herein, we develop a novel algorithm to efficiently compute the sparse Fourier transform of a generalized $q$-ary function. We can express any generalized $q$-ary function $f$ as: 
\begin{equation}
    f(\mathbf{m}) = \sum_{\mathbf{k} \in \mathbb{Z}_{\mathbf{q}}} F[\mathbf{k}] \prod_{i=1}^{\smash{n}} \omega_{q_i}^{m_i k_i}, \quad \mathbf{m} \in \mathbb{Z}_{\mathbf{q}},
    \label{eq:genq_f}
\end{equation}
where $\mathbf{q} = [q_1, q_2, \dots, q_n]$, $\mathbb{Z}_{\mathbf{q}} = \mathbb{Z}_{q_1} \times \mathbb{Z}_{q_2} \times \cdots \times \mathbb{Z}_{q_n}$ for brevity, and $\omega_{q_i} = e^{\frac{2\pi j}{q_i}}$. In particular, with $q_i = q, \forall i$, Equations (\ref{eq:qary_f}) and (\ref{eq:genq_f}) become identical. This formulation naturally operates directly on alphabets of different sizes.

\subsection{Contributions}
We summarize the main contributions of our work:
\begin{itemize}
    \item We develop GFast, an efficient algorithm that computes the sparse Fourier transform of a generalized $q$-ary function with a sample complexity of $O(Sn)$ and computational complexity of $O(Sn\log N)$ in the noiseless setting. The GFast software is available 
    on our GitHub repository.\footnote{\url{https://github.com/amirgroup-codes/GFast}}
    
    \item We develop a noise-robust version of GFast, dubbed NR-GFast, that computes the transform in the presence of additive Gaussian noise, with a sample complexity of $O(Sn^2)$ and computational complexity of $O(Sn^2 \log N)$.   

    \item We conduct large-scale synthetic experiments and demonstrate that GFast computes the sparse Fourier transform of generalized $q$-ary functions using up to $16\times$ fewer samples and runs up to $8\times$ faster than existing solutions. 

    \item We conduct real-world experiments to explain deep neural networks trained to predict heart disease diagnosis and protein fluorescence and demonstrate that NR-GFast computes the sparse Fourier transform using $13\times$ fewer samples compared to existing algorithms.
    
\end{itemize}

\subsection{Related Work}
Computing the sparse Fourier transform of signals has a rich history, with much of the focus directed toward developing efficient algorithms for the discrete Fourier transform (DFT)~\cite{pawar2017ffast, hassanieh2012nearly, amrollahi2019efficiently, hassanieh2012simple, indyk2014nearly, ghazi2013sample}. These algorithms leverage the principle of aliasing, where subsampling a signal in the time domain results in a linear mixing, or sketching, of the Fourier coefficients. Algorithms have been developed to efficiently sketch (encode) and recover (decode) Fourier coefficients. FFAST~\cite{ pawar2017ffast, pawar2017rffast} achieves this by inducing sparse graph alias codes in the DFT domain through a subsampling strategy guided by the Chinese Remainder Theorem and employing a belief-propagation decoder ~\cite{shokrollahi2004ldpc, ecc} to find the non-zero Fourier coefficients. Inspired by this coding-theoretic framework, algorithms are developed for the problem of computing the sparse Walsh Hadamard transform of pseudo-Boolean functions—a special case of a $q$-ary function with $q=2$~\cite{li2015spright, amrollahi2019efficiently, scheibler2015hadamard, cheraghchi2017walsh}. More recently, the $q$-SFT algorithm~\cite{erginbas2023efficiently} was developed to compute the sparse Fourier transform of a $q$-ary function for any integer $q$. 
GFast departs from all these previous works by directly subsampling in the $\mathbb{Z}_{\mathbf{q}}$ space. The resulting aliasing pattern induces a bipartite graph, which, unlike previous works, has an uneven degree distribution among nodes from different groups.
Our theoretical results in GFast demonstrate that the graph’s uneven degree distribution does not change the algorithm’s convergence guarantees in the asymptotic regime. GFast, in fact, computes the Fourier transform with a sample and time complexity that grows sublinearly in $N$, even in the presence of noise.

\section{Problem Setup}

Consider the generalized $q$-ary function $f({\bf m}):\mathbb{Z}_{\mathbf{q}}\rightarrow \mathbb{R}$ and let $N = \prod_{i =1}^nq_i$. Our goal is to find the Fourier transform of $f$ defined as,
\begin{equation}
F[\mathbf{k}] = \frac{1}{N} \sum_{\mathbf{m}\in\mathbb{Z}_{\mathbf{q}}}f(\mathbf{m})\prod_{i = 1}^n\omega_{q_i}^{-m_i k_i}, \quad 
    \mathbf{k} \in \mathbb{Z}_{\mathbf{q}},
\label{eq:inv_sitedep_fourier}
\end{equation}
when $F[{\bf k}]$ is $S$-sparse. To establish theoretical guarantees, we make the following assumptions.
\begin{assumption}
    \label{sublinear-sparsity}
    Let $\mathcal{S} \coloneqq \text{supp}(F)$ denote the support set of $F$ (i.e., the indices ${\bf k}$ of non-zero Fourier coefficients). We will make the following assumptions:
\begin{enumerate}
    \item Each element in $\mathcal{S}$ is uniformly distributed and chosen randomly across $\mathbb{Z}_{\mathbf{q}}$.
    \item The sparsity $S=|\mathcal{S}|$ is sub-linear in $N$, that is, $S = N^{\delta}$ for some $0 \leq \delta < 1$.
\end{enumerate}
\end{assumption}
\begin{assumption}
    \text{We also assume:} \\
    \vspace{-1.2em}
    \begin{enumerate}
        \label{random-coefficients}
        \item The Fourier coefficients $F[\mathbf{k}]$, for all $\mathbf{k} \in \mathcal{S}$, are sampled uniformly from the finite set $\mathcal{X} \coloneqq \{\rho, \rho \phi, \rho \phi^2 , \cdots , \rho \phi^{\kappa-1}\}$, where $\phi = e^{j\frac{2\pi}{\kappa}}$ for a constant phase offset $\kappa$ and strength factor $\rho$.
        \item The signal-to-noise ratio (SNR) is written as: $\text{SNR} = \frac{||f||^2}{N \sigma^2} = \frac{||F||^2}{\sigma^2} = \frac{S\rho^2}{\sigma^2},$
        and is assumed to be a constant.
        \item We are given access to noisy samples $f(\mathbf{m}) + w(\mathbf{m})$, where $w(\mathbf{m}) \sim \mathcal{C}\mathcal{N}(0, \sigma^2)$ is complex Gaussian noise.
    \end{enumerate}
\end{assumption}

\section{GFast: Noiseless}

We first analyze the case where samples drawn from $f$ are noiseless (i.e., $\sigma^2=0$). GFast finds the non-zero Fourier coefficients in three steps: 1) Subsampling and aliasing, 2) Bin detection, and 3) Peeling.

\subsection{Sampling and Aliasing}

GFast strategically aliases $f$ by computing small Fourier transforms using certain subsampling patterns. To do this, we create $C$ subsampling groups, each of size $b < n$, and take subsampled Fourier transforms over these samples. For each subsampling group $c$, we create the subsampling matrix $\mathbf{M}_c$ as an $n \times b$ matrix with the partial-identity matrix structure, where $b$ rows form an identity matrix and the remaining $n - b$ rows are composed entirely of zeros (see Section~\ref{appendix:example} for a simple example).
For each subsampling group, we additionally denote the length-$b$ vector $\mathbf{b}_c$ as a size $b$ subvector of $\mathbf{q}$,
\begin{equation}
    \label{alphabet-subset}
    \mathbf{b}_c \coloneqq \mathbf{M}_c^T \mathbf{q}, 
\end{equation}
and a set of $P$ offsets $\mathbf{d}_{c,p} \in \mathbb{Z}_{\mathbf{q}}$, where $p \in [P]$. For each $c,p$, the subsampled Fourier coefficients $U_{c,p}[\mathbf{j}]$ indexed by $\mathbf{j} \in \mathbb{Z}_{\mathbf{b}_c}$ are written as:
\begin{equation}
    U_{c,p}[\mathbf{j}] = \frac{1}{B_c} 
    \sum_{\bm{\ell} \in \mathbb{Z}_{\mathbf{b}_c}}f(\mathbf{M}_c\bm{\ell}+\mathbf{d}_{c,p})\prod_{i = 1}^b\omega_{(\mathbf{b}_c)_i}^{-j_i \ell_i}, 
    \quad 
        \mathbf{j} \in \mathbb{Z}_{\mathbf{b}_c},
\label{eq:sitedep_alias_function_p}
\end{equation}
where $B_c = \prod_{i = 1}^b (\mathbf{b}_c)_i.$
In Section \ref{aliasing-proof-section} of the Appendix, we show that this equals:
\begin{equation}
    U_{c, p}[\mathbf{j}] = \sum_{\mathbf{k}: \text{ }\mathbf{M}_c^T \mathbf{k}= \mathbf{j}} F[\mathbf{k}] \prod_{i = 1}^{n} \omega_{q_i}^{(\mathbf{d}_{c,p})_i k_{i}}, \quad 
    \mathbf{k} \in \mathbb{Z}_{\mathbf{q}}.
\label{eq:sitedep_alias_fourier_p}
\end{equation}
Let $\mathbf{s}_{c,\mathbf{k}} \in \mathbb{C}^{P}$ be the vector of offset signatures such that the $p^\text{th}$ value in $\mathbf{s}_{c,\mathbf{k}} $ is $\prod_{i = 1}^{n} \omega_{q_i}^{(\mathbf{d}_{c,p})_i k_{i}}$. By grouping the subsampled Fourier coefficients as $\mathbf{U}_c[\mathbf{j}] = [U_{c,1}[\mathbf{j}], \ldots, U_{c,P}[\mathbf{j}]]$ with their respective offsets stacked as a matrix $\mathbf{D}_c \in \mathbb{Z}_{\mathbf{q}}^{P \times n}$, we rewrite Equation~(\ref{eq:sitedep_alias_fourier_p}) as: 
\begin{equation}
    \mathbf{U}_{c}[\mathbf{j}] =  
    \sum_{\mathbf{k}: \text{ }\mathbf{M}_c^T \mathbf{k}= \mathbf{j}} F[\mathbf{k}] \mathbf{s}_{c,\mathbf{k}} .
\label{eq:sitedep_alias_fourier}
\end{equation}
In the noiseless implementation, we set $P = n$ and $\mathbf{D}_c = \mathbf{I}_{n \times n}$, where $\mathbf{I}_{n \times n}$ is an $n \times n$ identity matrix.%

\begin{algorithm}
\caption{GFast}
\label{alg}
\begin{algorithmic}[1]
\Require $b, C, P, \{\mathbf{M}_c\}_{c \in [C]}, \{\mathbf{D}_c\}_{c \in [C]}$
\State $\hat{F} \leftarrow \emptyset$
\For{each $c$ in $C$}  \Comment{Subsampling}
    \For{each $p$ in $P$}
        \State $\mathbf{}U_{c,p}[\mathbf{j}] = \frac{1}{B_c} \sum_{\bm{\ell} \in \mathbb{Z}_{\mathbf{b}_c}}f(\mathbf{M}_c\bm{\ell}+\mathbf{d}_{c,p})\prod_{i = 1}^b\omega_{(\mathbf{b}_c)_i}^{-j_i \ell_i}$
    \EndFor
\EndFor
\State $\mathcal{S} = \{(c, \mathbf{j}) : \text{Type}(\mathbf{U}_c[{\mathbf{j}}]) = \mathcal{H}_s(\hat{\mathbf{k}}, \hat{F}[\hat{\mathbf{k}}])\}$
\State $L = \{(c, \mathbf{j}) : \text{Type}(\mathbf{U}_c[{\mathbf{j}}]) = \mathcal{H}_m\}$
\While{$(|\mathcal{S}| > 0)$} \Comment{Peeling}
\For{each $c$ in $C$}
    \For{each $\mathbf{j}$ in $\mathbb{Z}_{\mathbf{b}_c}$}
        \If{$\text{Type}(\mathbf{U}_c[\mathbf{j}]) = \mathcal{H}_S (\hat{\mathbf{k}}, \hat{F}[\hat{\mathbf{k}}])$}
                \State $\hat{F} \leftarrow \hat{F} \cup (\hat{\mathbf{k}}, \hat{F}[\hat{\mathbf{k}})$
                \For{each $c'$ in $C$}
                    \State $\mathbf{j}' = \mathbf{M}_{c'}^T \hat{\mathbf{k}}$
                    \State $\mathbf{U}_{c'}[\mathbf{j'}] \leftarrow 
                \mathbf{U}_{c'}[\mathbf{j'}] - \hat{F}[\hat{\mathbf{k}}]\mathbf{s}_{c', \hat{\mathbf{k}}}$
                \State \text{Update }$\mathcal{S}, L$
                \EndFor
        \EndIf
    \EndFor
    \EndFor
\EndWhile
\State Return $\hat{F}$
\end{algorithmic}
\end{algorithm}

\subsection{Bin Detection and Recovery} 
\label{bindetection}
After aliasing $f$, we use the subsampling groups to recover the original Fourier coefficients. Each $\mathbf{U}_{c}[\mathbf{j}]$ is a linear combination of Fourier coefficients $F[\mathbf{k}]$. The Fourier coefficients are recovered using a bin detection procedure, where the objective is to identify which $\mathbf{U}_{c}[\mathbf{j}]$ contains only one Fourier coefficient.
To do this, in addition to the offset matrix $\mathbf{D}_c$, we choose a fixed delay $\mathbf{d}_{c,0} = \mathbf{0}_n$, where $\mathbf{0}_n$ is the vector of all $0$'s of length $n$. This means that $U_{c,0}[\mathbf{j}] = \sum_{\mathbf{k}:\text{ }\mathbf{M}_c^T \mathbf{k}= \mathbf{j}} F[\mathbf{k}]$. Using $U_{c,0}[\mathbf{j}]$, we are able to identify $\mathbf{U}_{c}[\mathbf{j}]$ as a \textit{zero-ton}, \textit{singleton}, or \textit{multi-ton} bin based on the following criteria: 
\begin{itemize}
    \item \textbf{Zero-ton verification:} $\mathbf{U}_{c}[\mathbf{j}]$ is a zero-ton (denoted by $\mathcal{H}_Z$) if $F[\mathbf{k}] = 0$ such that $\mathbf{M}_c^T \mathbf{k}= \mathbf{j}$ (the summation condition in Equation (\ref{eq:sitedep_alias_fourier})). In the presence of no noise, this is true when $U_{c,p}[\mathbf{j}] = 0$ for all $p = 1, \ldots, P$.
       \item \textbf{Singleton verification:} $\mathbf{U}_{c}[\mathbf{j}]$ is a singleton (denoted by $\mathcal{H}_S$) if there exists only one $\mathbf{k}$ where $F[\mathbf{k}] \neq 0$ such that $\mathbf{M}_c^T \mathbf{k}= \mathbf{j}$. This implies that $|U_{c,1}[\mathbf{j}]| = |U_{c,2}[\mathbf{j}]| = \ldots = |U_{c,P}[\mathbf{j}]|$, meaning that $\mathbf{U}_{c}[\mathbf{j}]$ is a singleton if:
    $$
    \left|\frac{U_{c,p}[\mathbf{j}]}{U_{c,0}[\mathbf{j}]}\right| = 1, \quad p = 1, 2, \ldots, P.
    $$
    \item \textbf{Multi-ton verification:} $\mathbf{U}_{c}[\mathbf{j}]$ is a multi-ton (denoted by $\mathcal{H}_M$) if there exists more than one $\mathbf{k}$ where $F[\mathbf{k}] \neq 0$ such that $\mathbf{M}_c^T \mathbf{k}= \mathbf{j}$. This means that:
    $$
    \left|\frac{U_{c,p}[\mathbf{j}]}{U_{c,0}[\mathbf{j}]}\right| \neq 1, \quad p = 1, 2, \ldots, P.
    $$
\end{itemize}
With this, we define the function $\text{Type}(\mathbf{U}_c[\mathbf{j}]) \rightarrow 
\{\mathcal{H}_Z, \mathcal{H}_M, \mathcal{H}_S(\mathbf{k}, F[\mathbf{k}]) \}$, which, given a bin, returns the type, along with a tuple of $(\mathbf{k}, F[\mathbf{k}])$ if the bin is a singleton. If a singleton is found, $U_{c,0}[\mathbf{j}] = F[\mathbf{k}]$. This enables us to use the identity structure of $\mathbf{D}_c$ to solve for $\mathbf{k}$:
\begin{equation}
    \begin{bmatrix}
        \text{arg}_{q_1}[U_{c,1}[\mathbf{j}]/U_{c,0}[\mathbf{j}]] \\
        \vdots \\
        \text{arg}_{q_n}[U_{c,P}[\mathbf{j}]/U_{c,0}[\mathbf{j}]]
    \end{bmatrix} = \mathbf{D}_c \mathbf{k},
\end{equation}
where $\text{arg}_{q_i} \colon \mathbb{C} \rightarrow \mathbb{Z}_{q_i}$  is the $q_i$-quantization of the argument of a complex number defined as,
\begin{equation}
    \label{quantization-formula}
    \text{arg}_{q_i}(z) := \left \lfloor \frac{q_i}{2\pi}\text{arg}
(ze^{\frac{j\pi}{q_i}}) \right \rfloor.
\end{equation}
After obtaining the unknown vector $\mathbf{k}$, the value of $F[\mathbf{k}]$ is directly obtained as $F[\mathbf{k}] = U_{c,0}[\mathbf{j}]$.

\subsection{Peeling Decoder}
\label{peelingdecoder}
Upon determining the bin type, the final step uses a peeling decoder to identify more singleton bins. This problem is modeled with a bipartite graph, where each $F[\mathbf{k}]$ is a variable node, each $\mathbf{U}_c[\mathbf{j}]$ is a check node, and there is an edge between $F[\mathbf{k}]$ and $\mathbf{U}_c[\mathbf{j}]$ when $\mathbf{M}_{c}^T \mathbf{k} = \mathbf{j}$. The graph is a left-$C$ regular sparse bipartite graph with a total of $SC$ nodes. Fig.~\ref{fig:bipartite_graph} shows an illustration of one such bipartite graph. The peeling algorithm is detailed in Algorithm \ref{alg}. Upon identifying a singleton, for every subsampling group $c$ in $C$, we subtract $F[\mathbf{\mathbf{k}}]$ from every $\mathbf{U}_c[\mathbf{M}_{c}^T \mathbf{k}]$. This is equivalent to peeling an edge off the bipartite graph. After looping through all subsampled Fourier coefficients, this procedure is repeated until no more singletons can be peeled. Under Assumptions 1 and 2, we prove the following Theorem: 
\begin{theorem}
    \label{noiseless}
    Given a function $f$ as Equation (\ref{eq:genq_f}) that satisfies Assumptions \ref{sublinear-sparsity} and \ref{random-coefficients}, with each $B_c = O(S)$, $C = O(1)$, $P = n$, and $\mathbf{D}_c = \mathbf{I}_{n \times n}$, there exists some set of $\mathbf{M}_c$ such that GFast recovers all Fourier coefficients in $\text{supp}(F)$ with probability at least $1 - O(1/S)$ with a sample complexity of $O(Sn)$ and computational complexity of $O(Sn\log N)$.
\end{theorem}
\begin{proof}
    The proof is in Section \ref{noiseless-full-proof} of the Appendix.
\end{proof}

\section{GFast: Noise Robust}

We develop a version of GFast that accounts for the presence of noise, where $\sigma^2 > 0$, dubbed NR-GFast. Instead of constraining $P = n$ and $\mathbf{D}_c = \mathbf{I}_{n \times n}$, given the hyperparameter $P_1$, we let $P = P_1 n$ and $\mathbf{D}_c$ be a matrix in $\mathbb{Z}_{\mathbf{q}}^{P \times n}$ where the entries of column $i$ are chosen uniformly at random over $\mathbb{Z}_{q_i}.$ For simplicity, we drop the subsampling group index $c$ and index $\mathbf{j}$ from $\mathbf{U}_{c}[\mathbf{j}]$ and write $\mathbf{U} = [\mathbf{U}_1, \cdots \mathbf{U}_p]^T$. For each group, let $\mathbf{S} = [\cdots \mathbf{s_k} \cdots] \in \mathbb{C}^{P \times N}$ be the matrix of offset signatures. In the presence of noise, $\mathbf{U} = \mathbf{S}\boldsymbol{\alpha} + \mathbf{W}$, where $\boldsymbol{\alpha}[\mathbf{k}] = F[\mathbf
k]$ if $\mathbf{M}^T{\mathbf{k}} = \mathbf{j}$, otherwise $\boldsymbol{\alpha}[\mathbf{k}] = 0$. $\mathbf{W}$ is a matrix of complex multivariate Gaussian noise with zero mean and covariance $\nu_c^2 \mathbf{I}$, such that $\nu_c^2 = \sigma^2/B_c$.  

NR-GFast uses random offsets that individually recover each $\mathbf{k}$ index based on repetition coding. Given a set of $P_1$ random offsets $\mathbf{d}_{p \in [P_1]}$ such that $P = P_1 n$, we modulate each offset with each column of the identity matrix such that we have $n$ offsets $\mathbf{d}_{p,r}:$
\begin{equation}
    \label{eq:nr-offset-formula}
    \mathbf{d}_{p,r} = \mathbf{d}_p \oplus_{q_r} \mathbf{e}_r, \forall p \in [P_1] , \quad \forall r \in n.
\end{equation}
This generates $C$ matrices $\mathbf{D}_c \in \mathbb{Z}_{\textbf{q}}^{P \times n}$. We then use a majority test to estimate the $r^\text{th}$ value of $\hat{\mathbf{k}}$, where $\hat{\mathbf{k}}$ is the estimated index of $\mathbf{k}$:
\begin{equation}
    \hat{k}_r = \argmax_{a \in \mathbb{Z}_{q_r}} \sum_{p \in [P_1]}\mathbbm{1}\{a = \text{arg}_{q_r}[U_{p,r}/U_p]\},
\end{equation}
where $\mathbbm{1}$ denotes the indicator function, and $\hat{k}_r$ denotes the value $a \leq q_r$ that is the maximum likelihood estimate of $q_r$. The value of the coefficient $\hat{F}[\hat{\mathbf{k}}]$ is then estimated with $
    \hat{F}[\hat{\mathbf{k}}] = \argmin_{\alpha \in \mathcal{X}} \| \alpha - \mathbf{s}_{\hat{\mathbf{k}}} ^T\mathbf{U}/P \|.$
\begin{figure}[t!]
    \vspace{-0.5cm}
    \centering
    \includegraphics[width=0.5\textwidth]{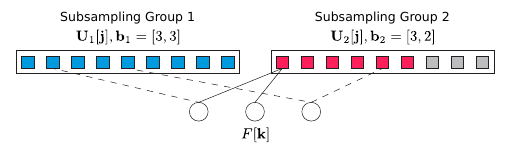}
    \caption{An example of a bipartite graph for computing the sparse Fourier transform of a generalized $q$-ary function with $n=4$, $b=2$, and $\mathbf{q} = [3, 3, 3, 2]$, where $\mathbf{b}_1 = [3, 3]$ (blue squares) and $\mathbf{b}_2 = [3, 2]$ (red squares). Dashed lines represent edges that are peeled in the first iteration of decoding in GFast. Existing algorithms for $q$-ary functions use the $\mathbf{b}_2 = [3, 3]$ (red and gray squares) grouping, which inflates the dimension. }
    \label{fig:bipartite_graph}
\end{figure} 
Then, for some $\gamma \in (0, 1)$, bin detection is modified using the following criteria: 
\begin{itemize}
    \item \textbf{Zero-ton verification:} $\text{Type}(\mathbf{U}_c[\mathbf{j}]) = \mathcal{H}_Z$ if $\frac{1}{P} \| \mathbf{U}_c[\mathbf{j}] \|^2 \leq (1 + \gamma) \nu_c^2$.
    \item \textbf{Singleton verification:} After ruling out zero-tons, we estimate $(\hat{\mathbf{k}}, \hat{F}[\hat{\mathbf{k}}])$, where $\hat{F}[\hat{\mathbf{k}}]$ is the estimated Fourier coefficient corresponding to $\hat{\mathbf{k}}$. We verify if $\mathbf{U}_c[\mathbf{j}]$ is a singleton if $\frac{1}{P} \| \mathbf{U}_c[\mathbf{j}] - \hat{F}[\hat{\mathbf{k}}] \mathbf{s}_{c, \hat{\mathbf{k}}}\|^2 \leq (1 + \gamma) \nu_c^2$.
    \item \textbf{Multi-ton verification:} If neither $\mathrm{Type}(\mathbf{U}_c[\mathbf{j}]) = \mathcal H_Z$ nor $\text{Type}(\mathbf{U}_c[\mathbf{j}]) = \mathcal{H}_S$, then $\text{Type}(\mathbf{U}_c[\mathbf{j}]) = \mathcal{H}_M$.
\end{itemize}
Using the above setup, the random offsets satisfy the following proposition: 
\begin{proposition}
    \label{proposition1}
    Given a singleton bin $(\mathbf{k}, F[\mathbf
    k])$, the $q_r$-quantized ratio as defined in Equation (\ref{quantization-formula}) for 
    \begin{equation}
    \begin{aligned}
        U_p &= F[\mathbf{k}]\prod_{i=1}^n \omega_{q_i}^{(\mathbf{d}_{p})_ik_i} + W_p, \\
        U_{p,r} &= F[\mathbf{k}]\prod_{i=1}^n \omega_{q_i}^{(\mathbf{d}_{p,r})_ik_i} + W_{p,r},
    \end{aligned}
\end{equation}
    satisfies $\text{arg}_{q_r} [U_{p,r}/U_p] = \langle \mathbf{e}_r, \mathbf{k}\rangle \otimes_{q_r} Z_{p, r}$, where $Z_{p, r}$ is a random variable over $\mathbb{Z}_{q_r}$ with a probability $p_i \coloneqq \mathbb{P}(Z_{p, r} = i)$. Let $\sum_{i \neq 0} p_i \coloneqq \mathbb{P}_e$. $\mathbb{P}_e$ is upper bounded by $2e^{-\frac{\zeta}{2}\text{SNR}}$ for $\zeta = \eta_c \sin^2(\pi/2q_r)$.
\end{proposition}
\begin{proof}
    The proof is in Section \ref{proposition1-proof} of the Appendix.
\end{proof}
Using Proposition \ref{proposition1}, we prove the following Theorem:
\begin{theorem}
    \label{noisy}
    Given a function $f$ as Equation (\ref{eq:genq_f}) that satisfies Assumptions \ref{sublinear-sparsity} and \ref{random-coefficients}, with each $B_c = O(S)$, $C = O(1)$, $P = n^2$, $\mathbf{D}_c$ defined as in Equation (\ref{eq:nr-offset-formula}), and if each $q_i = O(1) >1$, there exists some set of $\mathbf{M}_c$ such that NR-GFast recovers all Fourier coefficients in $\text{supp}(F)$ with probability at least  $1 - O(1/S)$ with a sample complexity of $O(Sn^2)$ and computational complexity of $O(Sn^2\log N)$.
\end{theorem}
\begin{proof}
    The proof is in Section \ref{noisy-full-proof} of the Appendix.
\end{proof}

\section{Experiments}

We first evaluate the performance of GFast on synthetic data by considering the scenario where four different generalized $q$-ary functions taking in an $n=20$ length sequence have $q_\text{max}=7$ (the distribution of alphabets is shown in the legend of Fig. \ref{fig:synt_figures}a). We run $q$-SFT on one function by setting $q_i = q_\text{max}, \forall i \in 1 \dots n$.
Following Assumptions \ref{sublinear-sparsity} and \ref{random-coefficients}, we synthetically generate an $S$-sparse Fourier transform $F$, where $\mathcal{S} = \text{supp}(F)$ is chosen uniformly at random in $\mathbb{Z}_{q_1=7} \times \mathbb{Z}_{q_2=3} \times \mathbb{Z}_{q_3=3} \times \mathbb{Z}_{q_4=3} \times \mathbb{Z}_{q_5=4} \times \cdots \times \mathbb{Z}_{q_{20} = 4}$  with values according to $F[\mathbf{k}] = \rho e^{-j \Omega[\mathbf{k}]} \text{ if } \mathbf{k} \in \mathcal{S}, \text{ else } 0$,
where $\Omega[\mathbf{k}]$ are independent and random variables sampled from $[0, 2\pi)$. We evaluate performance using normalized mean-squared error (NMSE), which is calculated as: $\text{NMSE} =  \frac{||\hat{F} - F||^2}{||F||^2} = \frac{||\hat{f} - f||^2}{||f||^2}$, where $\hat{f}$ is the estimated function using the recovered Fourier coefficients. In all experiments, we set $C = \floor{n/b}$.
Fig. \ref{fig:synt_figures}a and \ref{fig:synt_figures}b compare GFast and NR-GFast with the noiseless and noise-robust versions of $q$-SFT in terms of sample complexity, respectively. We utilize $q$-SFT's partial-identity subsampling matrix implementation to match the aliasing pattern derived in Equation (\ref{eq:sitedep_alias_fourier_p}). In the noiseless setting, $\rho=1$, $S=1000$, and $\sigma^2$ such that $\text{SNR}=100$dB. To vary the sampling patterns of $\mathbf{q}$ in the noiseless implementation of GFast, we additionally permute $\mathbf{q}$ and each vector of $\mathbf{k}$ five times per experimental instance. For both algorithms, we vary $b$ from 1 to 5 to get different sample counts. In the noisy setting, we set $\rho=1$, $S=1000$, and $\sigma^2$ such that $\text{SNR}=10$dB. We vary $b$ from 1 to 5 and $P_1$ from 18 to 20 in both algorithms. Each $q$-SFT and GFast experimental instance is run over three random seeds. When GFast achieves an NMSE of less than $0.01$, there is an average difference of $0.96 \pm 0.06$ and $0.92 \pm 0.06$ in NMSE at the closest sample complexity of $q$-SFT in the noiseless and noisy settings, respectively. When both GFast and $q$-SFT achieve an NMSE of less than $0.01$, GFast uses up to $16\times$ fewer samples and runs up to $8\times$ faster.

\begin{figure}[t!]
\centering
\includegraphics[width=0.5\textwidth]{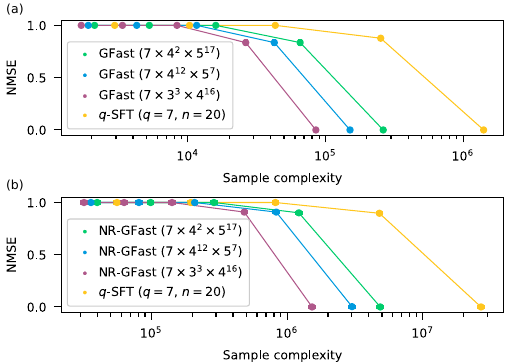}
\vspace{0cm}
\caption{NMSE of GFast compared to $q$-SFT on synthetically generated data, where the Fourier coefficients live in $\mathbb{Z}_{q_1=7} \times \mathbb{Z}_{q_2=3} \times \mathbb{Z}_{q_3=3} \times \mathbb{Z}_{q_4=3} \times \mathbb{Z}_{q_5=4} \times \cdots \times \mathbb{Z}_{q_{20} = 4}$, in the \textbf{(a)} noiseless and \textbf{(b)} noisy settings at various sample complexities. When achieving an NMSE of less than $0.01$, GFast uses up to $16\times$ fewer samples and runs up to $8\times$ faster than $q$-SFT.}
\vspace{-0.5cm} 
\label{fig:synt_figures}
\end{figure}

We additionally show the performance of NR-GFast in learning a real-world function. We train a multilayer perceptron $f$ to predict the presence of heart disease using data from~\cite{janosi1989heart}. $f$ takes $n=12$ categorical features as an input. Fig.~\ref{fig:tabular_figures}a shows the number of categories per feature, where the maximum alphabet size is $q=5$. Treating this problem as a $q$-ary function with $q=5$ artifically inflates the problem's dimensionality from $\prod_{i=1}^{n=12} q_i = 4.5 \times 10^{6}$ to $q^n = 5^{12} \sim 2.4 \times 10^{8}$. Moreover, evaluating $f$ on out-of-domain inputs $\mathbf{x} \in \mathbb{Z}_q^n \setminus \mathbb{Z}_{\mathbf{q}}$ unnecessarily increases sample complexity, since computing $f(\mathbf{x})$ requires freezing the input features $\mathcal{V} = \left\{ i \in \{1, \dots, n\} \mid x_i \in \mathbb{Z}_{q_i}] \right\}$ that exist in $\mathbb{Z}_{\mathbf{q}}$, and taking an expectation over the remaining samples:
$f(\mathbf{x}) = \mathbb{E}_{\mathbf{m} \sim \mathrm{Unif}(\mathbb{Z}_{\mathbf{q}})} \left[ f(\mathbf{m}) \,\middle|\, \mathbf{m}_{\mathcal{V}} = \mathbf{x}_{\mathcal{V}} \right]$. We run NR-GFast over inputs in $\mathbb{Z}_{\mathbf{q}}$ and $q$-SFT using $q=5$ and $n=12$, and calculate the test NMSE of both algorithms on $10,000$ samples chosen uniformly and randomly from $\mathbb{Z}_{\mathbf{q}}$. We vary $b$ from 1 to 3 and $P_1$ from $3$ to $12$ in both algorithms to get different sample complexities.

Fig. \ref{fig:tabular_figures}b compares NR-GFast with $q$-SFT. Despite NR-GFast taking $13\times$ fewer samples than $q$-SFT, NR-GFast achieves a minimal NMSE of $0.37$ compared to $q$-SFT achieving a minimal NMSE of $0.42$. Given a sample budget of $10^5$ samples, $q$-SFT achieves a significantly worse minimal NMSE ($0.53$) than NR-GFast. On an MLP trained to predict protein floresence (see Fig. \ref{fig:gfp_figures} in Section~\ref{appendix:results} in the Appendix), NR-GFast again outperforms $q$-SFT given a sample budget (with an NMSE of $0.42$ and $0.54$, respectively). Together, these results highlight the sample efficiency of NR-GFast in computing the sparse Fourier transform of generalized $q$-ary functions. 

\begin{figure}[t!]
\centering
\includegraphics[width=0.5\textwidth]{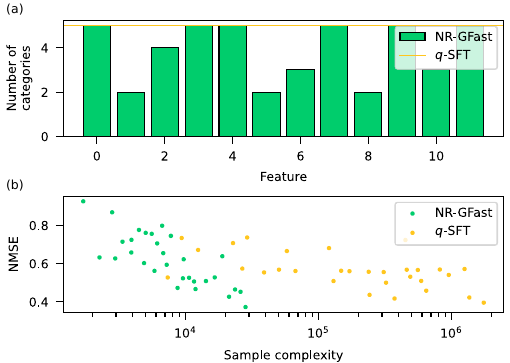}
\vspace{0cm}
\caption{\textbf{(a)} Distribution of categories per feature that comprise the generalized $q$-ary function in NR-GFast for $n=12$ compared to the $q$-ary function where $q=5$. \textbf{(b)} Performance of NR-GFast and $q$-SFT on the presence of heart disease as predicted by a trained MLP. NR-GFast uses up to $13\times$ fewer samples than $q$-SFT.}
\label{fig:tabular_figures}
\end{figure}

\section{Conclusion and Future Work}
In this paper, we introduced GFast, a fast and efficient algorithm for computing the sparse Fourier transform of generalized $q$-ary functions. Our theoretical analysis and empirical experiments demonstrate that GFast outperforms existing algorithms limited to $q$-ary functions by achieving faster computation and requiring fewer samples, especially when the alphabet size distribution is non-uniform. 

Several promising research directions remain to be explored. For example, our theoretical analysis with noise assumes bounded $q_i$, and our experiments are limited to alphabet sizes up to $12$. A significant increase in $q$, for example, to several thousand in language processing applications, will degrade GFast's performance. Developing algorithms that are robust to large alphabets is an important future challenge. Additionally, leveraging GFast to develop methods for explaining large-scale machine learning models of generalized $q$-ary functions—such as those used in biology and natural language processing—offers another exciting avenue for further research.

\balance
\bibliography{references.bib}
\bibliographystyle{unsrt}

\section{Additional Results on Protein Fluorescence}
\label{appendix:results}


We consider a multilayer perceptron (MLP) trained to predict fluorescence values of green fluorescence protein amino acid sequences~\cite{sarkisyan2016local}, where we aim to learn the function $f$ by mutating the first $n=12$ amino acids of the sequence. Using the weights of the first layer, where each amino acid is one-hot encoded as a vector of length 20 (for 20 amino acids), we identify the amino acids at each position most predictive of fluorescence by applying a threshold, which we set to 0.05. 
For a uniform alphabet, the vocabulary set $\mathcal{A}$ can be written as $\{\alpha_1, \alpha_2, \cdots , \alpha_n \}$ where $|\mathcal{A}| = q$ for all $i.$ However, given the weights $\mathbf{W}^{n \times q}$ of the first layer of the fitness model, the vocabulary set can be found for each site as $$\mathcal{A}_i = \{\alpha \in A: \mathbf{W}_{i, a} \geq \tau \}, |\mathcal{A}_i| = q_i,$$ where we set $\tau = 0.05$. 

Fig. \ref{fig:gfp_figures}a shows the distribution of predictive amino acids, where the maximum alphabet size is $q=12$. Similar to the heart disease problem, treating this problem as a generalized $q$-ary function allows us to reduce the dimensionality of the problem from $q^n =12^{12}\sim8.9 \times 10^{12}$ down to $8.6 \times 10^{10}$. We run NR-GFast on the most predictive amino acids in $\mathbb{Z}_{\mathbf{q}}$ and $q$-SFT using $q=12$ and $n=12$, again calculating the test NMSE of both algorithms on a $10,000$ samples chosen uniformly and randomly from $\mathbb{Z}_{\mathbf{q}}$. We vary $b$ from 1 to 3 and $P_1$ from $5$ to $20$ in both algorithms to get different sample complexities.

In Fig. \ref{fig:gfp_figures}b, given a sample budget of $10^6$ samples, NR-GFast achieves a minimal NMSE of $0.42$, while $q$-SFT achieves a minimal NMSE of $0.54$. Despite $q$-SFT running at sample complexities up to $10^7$, the minimal NMSE ($0.49$) is still worse than NR-GFast. These results highlight the efficiency of NR-GFast in problem setups where leveraging the predictive alphabet space in $f$ can reduce sample complexity.

\begin{figure}[h]
\centering
\includegraphics[width=0.5\textwidth]{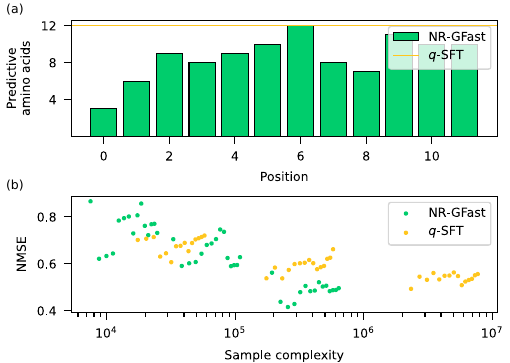}
\vspace{0cm}
\caption{\textbf{(a)} Distribution of predictive amino acids per position that comprise the generalized $q$-ary function in NR-GFast for $n=12$ compared to the $q$-ary function where $q=12$. \textbf{(b)} Performance of NR-GFast and $q$-SFT on florescence in protein sequences as predicted by a trained MLP. Given a sample budget of $10^6$, NR-GFast achieves a lower minimal NMSE ($0.42$) than $q$-SFT ($0.54$).}
\vspace{-0cm} 
\label{fig:gfp_figures}
\end{figure}
\section{Appendix} 
\subsection{Simple Example}
\label{appendix:example}
We outline a simple example of the subsampling and bin detection stages of GFast on an example function. Consider the following function with $\mathbf{q} = [2, 2, 2, 3]$ that takes in a sequence of length $n=4$: $f(\mathbf{m}) = \sum_{\mathbf{k} \in \mathbb{Z}_{\mathbf{q}}} F[\mathbf{k}] \prod_{i=1}^{\smash{4}} \omega_{q_i}^{m_i k_i}, \quad \mathbf{m} \in \mathbb{Z}_{2} \times \mathbb{Z}_2 \times \mathbb{Z}_2 \times \mathbb{Z}_3,$ such that $N = 2 \times 2 \times 2 \times 3 = 24.$ We will set $b = 2$ and $C = 2$ using the following subsampling matrices: $$\mathbf{M}_1 = \begin{bmatrix}
    1 & 0 \\ 0 & 1 \\ 0 & 0 \\ 0 & 0
\end{bmatrix}, \mathbf{M}_2 = \begin{bmatrix}
    0 & 0 \\ 0 & 0 \\ 1 & 0 \\ 0 & 1
\end{bmatrix}.$$
This gives us $\mathbf{b}_1 = [2, 2]$ and $\mathbf{b}_2 = [2, 3].$ For simplicity, we drop the notation $p$ in $U_{c,p}$, and assume $\mathbf{d}_{c,p}$ is a vector of all zeros. We can now compute the subsampled transforms as: $$U_1[\mathbf{j}] = \frac{1}{2 \times 2}\sum_{\boldsymbol{\ell} \in \mathbb{Z}_{\mathbf{b}_1}} f(\ell_1 \ell_2 0 0)\prod_{i = 1}^2\omega_{(\mathbf{b}_1)_i}^{-j_i \ell_i}, \quad U_2[\mathbf{j}] = \frac{1}{2 \times 3}\sum_{\boldsymbol{\ell} \in \mathbb{Z}_{\mathbf{b}_2}} f(0 0 \ell_1 \ell_2)\prod_{i = 1}^2\omega_{(\mathbf{b}_2)_i}^{-j_i \ell_i}.$$ Using Equation (\ref{eq:sitedep_alias_fourier_p}), we can write the subsampled transforms in terms of the linear sums of their Fourier coefficients: \begin{equation}
\begin{aligned}
    \label{eq: big-system-equations}
    U_1[00] &= \sum_{\mathbf{k}:\text{ } \mathbf{M}^T\mathbf{k} = \mathbf{j}} F[\mathbf{k}] = F[0000] + F[0001] + F[0002] + F[0010] + F[0011] + F[0012] \\
    U_1[01] &= \sum_{\mathbf{k}:\text{ } \mathbf{M}^T\mathbf{k} = \mathbf{j}} F[\mathbf{k}] = F[0100] + F[0101] + F[0102] + F[0110] + F[0111] + F[0112] \\
    &\vdots \\
    U_2[12] &= \sum_{\mathbf{k}:\text{ } \mathbf{M}^T\mathbf{k} = \mathbf{j}} F[\mathbf{k}] = F[0012] + F[0112] + F[1012] + F[1112].
\end{aligned}
\end{equation} 
Let us assume that $F$ is $S$-sparse with $S = 5,$ where the Fourier coefficients $F[0100], F[1000], F[0112], F[1011],$ and $F[1111]$ are all non-zero. We can rewrite Equation (\ref{eq: big-system-equations}) as: \begin{equation}
    \begin{aligned}
        U_1[01] = F[0100] + F[0112] \\ U_1[10] = F[1000] + F[0112] \\ U_1[11] = F[1111] \\ U_2[00] = F[0100] + F[1000] \\ U_2[11] = F[1011] + F[1111] \\ U_2[12] = F[0112].
    \end{aligned}
\end{equation}
We will now introduce $\mathbf{d}_{c,p}$. Since our delays are based on the identity matrix, for $c = 1$ and $2$ we have the following offsets: $$\mathbf{d}_{c, 0} = \mathbf{0}_n, \mathbf{d}_{c, 1} = [1, 0, 0, 0]^T, \quad \mathbf{d}_{c, 2} = [0, 1, 0, 0]^T, \quad \mathbf{d}_{c, 3} = [0, 0, 1, 0]^T, \quad \mathbf{d}_{c, 4} = [0, 0, 0, 1]^T.$$ This allows us to write the bin vectors $\mathbf{U}_c[\mathbf{j}]$ as shifted versions of the corresponding Fourier coefficients found in the bin. For example, we can rewrite the multi-ton $\mathbf{U}_1[01]$ as: $$\mathbf{U}_1[01] = F[0100]\begin{bmatrix}
    1 \\ 
    \omega_2^{(1)(1)} = -1 \\
    \omega_2^{(1)(0)} = 1 \\
    \omega_2^{(1)(0)} = 1 \\
    \omega_2^{(1)(0)} = 1 
\end{bmatrix} + F[0112] \begin{bmatrix}
    1 \\ \omega_2^{(1)(0)} = 1 \\ \omega_2^{(1)(1)} = -1 \\ \omega_2^{(1)(1)} = -1 \\ \omega_3^{(1)(2)} = e^{4j\pi/3}
\end{bmatrix}.$$ 
 We can also write the singleton bin vector $\mathbf{U}_2[12]$ as: $$\mathbf{U}_2[\mathbf{j}] = F[0112] \begin{bmatrix}
    1 \\ \omega_2^{(1)(0)} = 1 \\ \omega_2^{(1)(1)} = -1 \\ \omega_2^{(1)(1)} = -1 \\ \omega_3^{(1)(2)} = e^{4j\pi/3}
\end{bmatrix}.$$ Using the noiseless bin detection outlined in Section \ref{bindetection}, we find the index $\mathbf{k}$ as: $$\mathbf{D}_c\mathbf{k} = \mathbf{k} = \begin{bmatrix}
    \text{arg}_2(1) = 0 \\ \text{arg}_2(-1) = 1\\ \text{arg}_2(-1) = 1\\ \text{arg}_3(e^{4j\pi/3}) = 2
\end{bmatrix}.$$
Repeating this process for all singletons will allow us to begin the peeling operation as detailed in Section \ref{peelingdecoder}.
\subsection{Proof of Equation (\ref{eq:sitedep_alias_fourier_p})}
\label{aliasing-proof-section}
\begin{proof}
    Let us first examine the case with no delay:
    \begin{equation}
        \label{initial-aliasing-proof}
        U_{c,p}[\mathbf{j}] = \frac{1}{B_c} \sum_{\boldsymbol{\ell} \in \mathbb{Z}_{\mathbf{b}_c}} f(\mathbf{M}_c\boldsymbol{\ell}) \prod_{m=1}^b \omega_{\left(\mathbf{b}_c\right)_i}^{-j_m \ell_m}.
    \end{equation}
    Using Equations (\ref{eq:inv_sitedep_fourier}) and (\ref{initial-aliasing-proof}), this is equivalent to: 
    \begin{equation*}
        \label{expansion-aliasproof}
         U_{c,p}[\mathbf{j}] = \frac{1}{B_c} \sum_{\boldsymbol{\ell} \in \mathbb{Z}_{\mathbf{b}_c}} \left(\sum_{\mathbf{k} \in \mathbb{Z}_{\mathbf{q}}} F[\mathbf{k}] \prod_{i=1}^n \omega_{q_i}^{k_i (\mathbf{M}_c\boldsymbol{\ell})_i}\right) \prod_{m=1}^b \omega_{\left(\mathbf{b}_c\right)_i}^{-j_m \ell_m}.
    \end{equation*}
    Using the structure of $\mathbf{M}_c$ from Equation (\ref{thiscasematrices}), we know that $\mathbf{M}_c$ will freeze an $n-b$ bit segment of $\mathbf{M}_c\boldsymbol{\ell}$ to 0, which lets us rewrite as: 
    \begin{align*}
            U_{c,p}[\mathbf{j}] &= \frac{1}{B_c}  \sum_{\mathbf{k} \in \mathbb{Z}_{\mathbf{q}}} F[\mathbf{k}] \sum_{\boldsymbol{\ell} \in \mathbb{Z}_{\mathbf{b}_c}} \prod_{m=1}^b \omega_{\left(\mathbf{b}_c\right)_i}^{\ell_m((\mathbf{M}_c^T\mathbf{k})_m -j_m)} \\  &=  \sum_{\mathbf{k}:  \mathbf{M}_c^T\mathbf{k} = \mathbf{j}} F[\mathbf{k}],
    \end{align*}
    where the last step uses Lemma \ref{lemma: zero-vector-innerprod}. We can then use the shifting property of the Fourier transform to get the pattern with the delay: \begin{equation}
    U_{c, p}[\mathbf{j}] = \sum_{\mathbf{k}: \text{ }\mathbf{M}_c^T \mathbf{k}= \mathbf{j}} F[\mathbf{k}] \prod_{i = 1}^{n} \omega_{q_i}^{(\mathbf{d}_{c,p})_i k_{i}}.
    \end{equation} 
\end{proof}
\begin{lemma}
    \label{lemma: zero-vector-innerprod}
    For some vector $\mathbf{a} \in \mathbb{Z}_{\mathbf{q}} = \mathbb{Z}_{q_1} \times \mathbb{Z}_{q_2} \times \cdots \mathbb{Z}_{q_n} $, 
    \begin{equation}
        \label{eq:zero-vector-alias-lemma}
        \sum_{\boldsymbol{\ell} \in \mathbb{Z}_{\mathbf{q}}}\prod_{i=1}^n \omega_{q_i}^{\ell_i a_i} = 0 \Leftrightarrow \mathbf{a} \neq \mathbf{0}_n.
    \end{equation}
\end{lemma}
\begin{proof}
    If we let $\mathbf{a} = \mathbf{0}_n$, then the product term becomes: $$\prod_{i=1}^n \omega_{q_i}^{\ell_i(0)}.$$ This product reduces to 1, which simplifies Equation (\ref{eq:zero-vector-alias-lemma}) to: $$\sum_{\boldsymbol{\ell} \in \mathbb{Z}_{\mathbf{q}}} 1 = q_1q_2\cdots q_n.$$
    For $\mathbf{a} \neq \mathbf{0}_n$, we can then rewrite Equation (\ref{eq:zero-vector-alias-lemma}) as: $$\prod_{i=1}^n \sum_{\ell_j \in \mathbb{Z}_{q_{i}}} \omega_{q_i}^{\ell_j a_i}.$$ 
    The summation term is a sum over all of the roots of unity of $\omega_{q_i}$, which comes out to 0. This makes the entire expression equal 0.
 \end{proof}

\subsection{Proof of Theorem \ref{noiseless}}
\label{noiseless-full-proof}
Our proof follows a similar structure to those seen in other papers that analyze low-density parity check (LDPC) codes. First, we find the general distribution of check nodes in the graph and show that the expected number of Fourier coefficients not recovered converges to 0 given sufficient samples (Lemmas \ref{lemma: poisson-distribution-lemma} and \ref{lemmma: bound conditional mean}). Then, we use a martingale argument to show that the number of Fourier coefficients not recovered converges to its mean (Lemmas \ref{azuma-finite-difference}, \ref{azuma final}, and \ref{tree-neighborhood}). Lastly, we show that peeling is always possible by showing that any subset of the variable nodes makes an expander graph (Lemma \ref{expander probability}). We use the matrix constructions for the case of $0 \leq \delta \leq 1/3,$ but similar constructions can be found for the less sparse regime in Section B.3 of \cite{li2015spright}. We assume that the subsampling matrices $\mathbf{M}_c$ are constructed as 
\begin{equation}
\label{thiscasematrices}
    \mathbf{M}_c = [\mathbf{0}_{b \times b(c-1)}, \mathbf{I}_{b \times b}, \mathbf{0}_{b \times (n - cb)}]^T.
\end{equation}

\begin{lemma}
    \label{lemma: poisson-distribution-lemma}
    Let $c$ be an arbitrary subsampling group where the alphabet subset is defined as seen in Equation (\ref{alphabet-subset}). Then, the fraction of edges connected to check nodes in $c$ with degree $j$ can be written as:
    \begin{equation}
        \label{degree-distribution}
        \rho_{c, j} = \frac{(1/\eta_c)^{j-1} e^{-1/\eta_c}}{(j-1)!} \quad j = 0, \cdots, S
    \end{equation}
\end{lemma}
\begin{proof}
    Let $\mathcal{G}(S, \{\eta_{c}\}_{c \in [C]}, C, \{{\mathbf{M}}_{c}\}_{c \in [C]})$ represent the set of all bipartite graphs from subsampling matrices $\mathbf{M}_c$. We again assume that the subsampling matrices are constructed according to Equation (\ref{thiscasematrices}).
    This lets us assume that for a certain $c$ where $\mathbf{j} = \mathbf{M}_c\mathbf{k}$, each $\mathbf{j}_c$ is chosen uniformly and randomly from $\mathbb{Z}_{\mathbf{b}_c}$.
    For each group $c$, let $\eta_c$ be a redundancy parameter such that $B_c = S\eta_c$.  The proof for this closely follows Appendix B.1 of~\cite{li2015spright}.
    except we analyze the degree distribution of the bipartite graph for each group $c$, rather than the entire graph, because the bipartite graph is biased by the varied alphabet size and size of $B_c$ for each group. 
    For a single group $c$, with total number of hashed coefficients $B_c$, we have $S$ edges. This means that for a certain group $c$, the fraction of edges connecting to check nodes with degree $j$ is:
    \begin{align*}
        \rho_{c, j} = \frac{\text{Pr(check node has degree }j)B_cj}{S} \\ = \text{Pr(check node has degree }j)\eta_cj.
    \end{align*}
     Since the probability of some arbitrary edge $e$ being connected to a check node $\mathbf{j}$ is $1/B_c$, the probability that $\mathbf{j}$ has degree $j$ can be found with a binomial distribution: $Binom(\frac{1}{\eta_c}S, S)$. We can approximate this by a Poisson distribution: $$\text{Pr(check node has degree }j) = \frac{(1/\eta_c)^j e^{-1/\eta_c}}{j!}.$$ This gives us the result seen in Equation (\ref{degree-distribution}): $$\rho_{c, j} = \frac{(1/\eta_c)^{j-1} e^{-1/\eta_c}}{(j-1)!}$$ 
\end{proof}
\begin{figure}[h!]  
    \vspace*{0in}   
    \hspace*{1.8in}  
    \scalebox{0.9}{ 
        \begin{tikzpicture}[
    box/.style={draw, minimum size=0.7cm},
    node/.style={draw, circle, minimum size=0.8cm, fill=white},
    scale=1.2, transform shape
]
    \node[box, fill=violet!40] (v) at (0,2.5) {$c$};
    
    \node[node, fill=blue!10] (c) at (0,0) {$v$};
    
    \node[box, fill=blue!40] (b1) at (-2,-2) {};
    \node[box, fill=yellow!20] (b2) at (0,-2) {$c'$};
    \node[box, fill=green!20] (b3) at (2,-2) {};
  
    \node[node, fill=blue!10] (c1a) at (-3,-4) {};
    \node[node, fill=blue!10] (c1b) at (-1,-4) {};
   
    \node[node, fill=blue!10] (c2) at (0,-4) {};

    \node[node, fill=blue!10] (c3a) at (1,-4) {};
    \node[node, fill=blue!10] (c3b) at (2,-4) {};
    \node[node, fill=blue!10] (c3c) at (3,-4) {};

    \node[right] at (3,-2) {$C-1=3$ check nodes};

    \node[right] at (3.5,-4) {};
  
    \draw (v) -- node[right] {} (c);

    \draw (c) -- (b1);
    \draw (c) -- (b2);
    \draw (c) -- (b3);
  
    \draw (b1) -- (c1a);  
    \draw[dotted, line width=1pt, dash pattern=on 2pt off 2pt] (b1) -- (c1b);  
    \draw[dotted, line width=1pt, dash pattern=on 2pt off 2pt] (b2) -- (c2);  
    \draw (b3) -- (c3a); 
    \draw (b3) -- (c3b);  
    \draw[dotted, line width=1pt, dash pattern=on 2pt off 2pt] (b3) -- (c3c);  
\end{tikzpicture}
    }
    \caption{Directed neighborhood $\mathcal{N}_e^{2i}$ of depth 2 for an edge $e$ over a particular graph $\mathcal{G}(S, \{\eta_{c}\}_{c \in [C]}, C, \{{\mathbf{M}}_{c}\}_{c \in [C]})$ with $C=4$ subsampling groups. Dotted lines represent edges peeled on the last iteration. $c$ could be peeled from the graph on the next iteration because check node $c'$ now has degree 1.}
    \label{fig:neighborhood}
\end{figure}
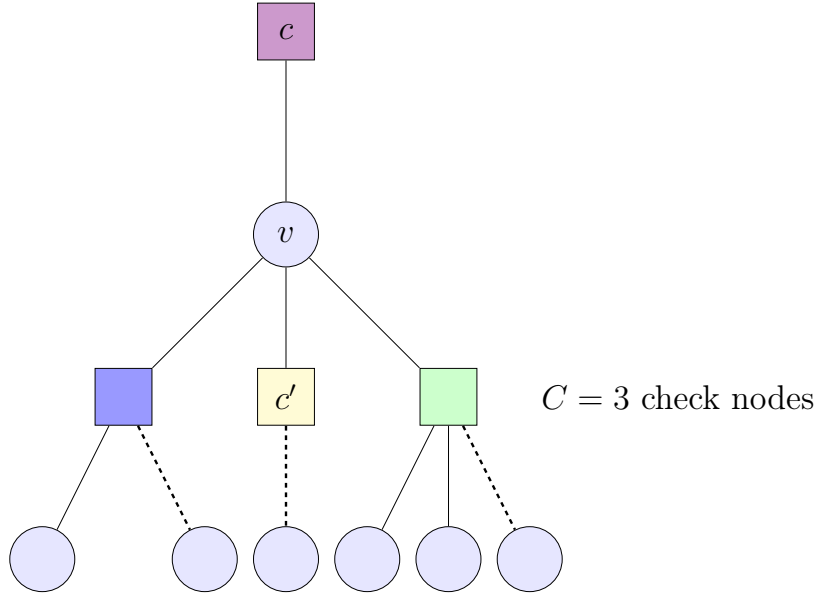

For the next step, we will analyze the density evolution of the peeling decoder~\cite{guruswami2006ldpc}.
Let $\mathcal{N}_{e}^{\ell}$ be the \textit{directed neighborhood} of edge $e$ with depth $\ell$, which can be written as $e_1, e_2, \cdots e_{\ell}, e \neq e_1$.


Given the local cycle-free neighborhood $\mathcal{N}_{e}^{2i}$ of an edge $e = (v, c')$, we can find the probability $p_{c,i}$ that $e$ is present after $i$ iterations as: \begin{equation}
    \label{eq:coupon}
    p_{c', i} = \prod_{c \in [C] \setminus \{c'\}} \left(1 - \sum_j \rho_{c, j} (1-p_{c, i-1})^{j-1} \right),
\end{equation} 
where $[C] = \{1,\ldots,C\} $. This can further be approximated as: \begin{equation}
    \label{eq:coupon_approx}
    p_{c', i} = \prod_{c \in [C] \setminus \{c'\}} \left(1 - e^{\frac{1}{\eta_c}p_{c, i-1}} \right),
\end{equation} 
due to the Poisson distribution of the check nodes~\cite{shokrollahi2004ldpc}. It is possible to choose $\{\eta_c\}_{c \in [C]}$ such that $p_{c, i}$ converges to 0 for all $c$. Table 1 of~\cite{li2015spright} gives the values of $\eta$ needed for convergence for a certain value of $C$. In our case, it is possible to find similar values of $\eta_c$ that cause convergence.

Then, let $\mathcal{T}_i$ be the event that for every edge $e$ the neighborhood $\mathcal{N}_{e}^{2i}$ is a tree, and let $Z_i$ be the number of edges that are still not peeled after iteration $i$, such that $Z_i = \sum_{c = 1}^C \sum_e Z_i^{e_{c'}}$, where $Z_{i}^{e_{c'}} = \mathbbm{1}\{\text{Edge }e\text{ in group } c' \text{ is not peeled}\}$. Then, the expected number of edges not peeled can be written as:
\begin{equation}
    \label{eq:peelingexpectation}
    \mathbbm{E}[Z_i | \mathcal{T}_i] = S\sum_{c=1}^C p_{c, i}.
\end{equation}
This shows that, given that the neighborhoods are all trees, we can terminate the peeling operation for an arbitrarily small number of edges left: $Z_i < \epsilon S\sum_{c=1}^Cp_{c,i} $. We show in Lemma \ref{tree-neighborhood} that the probability all neighborhoods $\mathcal{N}_e^{2i}$ are trees increases as $S$ increases. We can show that the difference between the expectation and the expectation conditioned on the tree assumption can be made arbitrarily small: 
\begin{lemma} (Adapted from Lemma 10 of~\cite{pedarsani2017phasecode})
\label{lemmma: bound conditional mean}
    For some $\epsilon > 0$ and sufficiently large $S$, there exists a peeling iteration $i$ that satisfies $$|\mathbb{E}[Z_i] - \mathbb{E}[Z_i|\mathcal{T}_i]| \leq SC\epsilon/2.$$
\end{lemma}
\begin{proof}
    We denote $Z_i^c$ as $SZ_i^{e_c},$ or the number of edges connected to check nodes in group $c$ on iteration $i.$ Due to linearity of expectation, we can write $\mathbb{E}[Z_i^c] = S\mathbb{E}[Z_i^{e_c}].$ 
   
    We can then use total expectation to rewrite $\mathbb{E}[Z_i^{e_c}]$: \begin{equation}
        \mathbb{E}[Z_i^{e_c}] = \mathbb{E}[Z_i^{e_c} | \mathcal{T}_i]\text{Pr}(\mathcal{T}_i) + \mathbb{E}[Z_i^{e_c} | \mathcal{T}_i^c]\text{Pr}(\mathcal{T}_i^c). 
    \end{equation}
    Trivially, we know that $\mathbb{E}[Z_i^{e_c} | \mathcal{T}_i^c] \leq 1$. We also know that $\mathbb{E}[Z_i^e | \mathcal{T}_i] = p_{c,i}$. We can bound $\mathbb{E}[Z_i^c]$ as $$\mathbb{E}[Z_i^{e_c}|\mathcal{T}_i]\text{Pr}(\mathcal{T}_i) \leq \mathbb{E}[Z_i^{e_c}] \leq \mathbb{E}[Z_i^{e_c} | \mathcal{T}_i] + \text{Pr}(\mathcal{T}_i^c).$$
    From Lemma \ref{tree-neighborhood}, we know that $\text{Pr}(\mathcal{T}_i) = 1 - O(\log^i(S)/S)$, which we can rewrite as $1 - \gamma \frac{\log^i(S)}{S}$ for some constant $\gamma$. With a little algebra, we can get: $$-\gamma\log^i(S) \leq \mathbb{E}[Z_i^c] - \mathbb{E}[Z_i^c | \mathcal{T}_i] \leq \gamma\log^i(S),$$ where $\mathbb{E}[Z_i^c | \mathcal{T}_i] = Sp_{c, i},$ which simplifies to $|\mathbb{E}[Z_i^c] - \mathbb{E}[Z_i^c | \mathcal{T}_i]| \leq S\epsilon/2$ if $\frac{S}{\log^i (S)} \geq 2\gamma/\epsilon.$ Due to the linearity of expectation, we also know that $$\sum_{c = 1}^C (\mathbb{E}[Z_i^c] - \mathbb{E}[Z_i^c | \mathcal{T}_i]) = \mathbb{E}[Z_i] - \mathbb{E}[Z_i | \mathcal{T}_i],$$ which gives us the result: $$|\mathbb{E}[Z_i] - \mathbb{E}[Z_i | \mathcal{T}_i]| \leq SC\epsilon/2.$$ 
\end{proof}
Now that we have shown that we can make the mean arbitrarily small, we would like to show that the number of unpeeled edges will converge to its mean value. To show convergence around the mean, we will set up a martingale $Y_0, \cdots Y_{SC}$, where $Y_{\ell} = \mathbb{E}[Z_i | Z_1^i, Z_2^i, \cdots Z_{\ell}^i]$ such that $Y_0 = \mathbb{E}[Z_i]$ and $Y_{SC} = \mathbb{E}[Z_i | Z_1^i, Z_2^i \cdots Z_{SC}^i] = Z_i$. 
This is a Doob's martingale process, which means we can apply Azuma's inequality: \begin{equation}
    \label{azumadefinition}
    \text{Pr}(|Z_i - \mathbb{E}[Z_i]| > t ) \leq 2\exp\left(-\frac{t^2}{2\sum_{i=1}^{SC}\alpha_i^2}\right).
\end{equation} We must first find the finite difference $|Y_i - Y_{i-1}| = \alpha_i$. It was shown in Theorem 2 of~\cite{richardson2001capacity} that, for a Tanner graph with regular variable degree $d_V$ and regular check degree $d_C$, $\alpha_i$ would be upper bounded by $8(d_Vd_C)^i$. However, as seen in Equation (\ref{degree-distribution}), the check node degree distribution of $\mathcal{G}(S, \{\eta_{c}\}_{c \in [C]}, C, \{{\mathbf{M}}_{c}\}_{c \in [C]})$ is not uniform, so we must adjust for the irregular degrees of the check node.

\begin{lemma}
    \label{azuma-finite-difference}
    (Adapted from~\cite{li2015spright}) Using the definition of Azuma's inequality given in Equation (\ref{azumadefinition}), we can bound the finite martingale difference as $\alpha_i^2 = O(S^{\frac{4i}{4i+1}})$.
\end{lemma} 

\begin{proof}
    We will first find the probability of $d_C \geq O(S^{\frac{2}{4i+1}})$. Then, for some constants $d_1, d_2, \beta$, we can bound this probability first using a Chernoff bound for a Poisson random variable $X$ with mean $\lambda$: 
\begin{equation}
    \label{eq:chernoff-part1}
    \text{Pr}(X > d_1 S^{\frac{2}{4i+1}}) \leq \left(\frac{e \lambda}{d_1 S^{\frac{2}{4i+1}}}\right)^{d_1 S^{\frac{2}{4i+1}}}.
\end{equation}
Then, we can bound the right-hand side as:
\begin{equation}
    \label{eq: chernoff-part2}
    \left(\frac{e \lambda}{d_1 S^{\frac{2}{4i+1}}}\right)^{d_1 S^{\frac{2}{4i+1}}} \leq
    d_2 \exp(-\beta S^{\frac{2}{4i+1}}).
\end{equation}
Note that this applies for all groups $c$ since every group has Poisson distributed check node degrees.
Now, let $\mathcal{B}$ be the event that at least one check node has a degree more than $O(S^{\frac{2}{4i+1}})$. Given the result of Equation (\ref{eq: chernoff-part2}) and using a union bound across all the nodes in $B_c = \eta_c S$, we can write: $$\text{Pr}(\mathcal{B}) < d_3S\exp(-\beta S^{\frac{2}{4i+1}}).$$
Let $\mathcal{B}^c$ be the event that no check node has a degree more than $O(S^{\frac{1}{4i+1}})$. If $\mathcal{B}$ has not happened, then we can bound $d_C$ by, $$d_C \leq O(S^{\frac{2}{4i+1}}),$$ and get a finite difference, $$\alpha_i^2 = (8(d_Vd_C)^i)^2 = O(S^{\frac{4i}{4i+1}}).$$
\end{proof} With this result we can look closer at the concentration around the mean:
\begin{lemma} (Adapted from~\cite{li2015spright})
\label{azuma final}
    For some constant $\beta$, $\text{Pr}(|Z_i - \mathbb{E}[Z_i]| > SC\epsilon/2 ) \leq 2 \exp(-\beta \epsilon^2 S^{\frac{1}{4i+1}}).$
\end{lemma}
\begin{proof}
    Using the law of total probability, we can rewrite the probability as:
    $$
        \text{Pr}(|Z_i - \mathbb{E}[Z_i]| > CS\epsilon/2 ) = \text{Pr}(|Z_i - \mathbb{E}[Z_i]| > CS\epsilon/2 \ | \ \mathcal{B})\text{Pr}(\mathcal{B}) + \text{Pr}(|Z_i - \mathbb{E}[Z_i]| > CS\epsilon/2 \ | \ \mathcal{B}^c)\text{Pr}(\mathcal{B}^c),$$ which can once again be bounded with, \begin{equation}
            \label{total-prob-inequality-finite}
            \text{Pr}(|Z_i - \mathbb{E}[Z_i]| > CS\epsilon/2 )\leq \text{Pr}(\mathcal{B}) + \text{Pr}(|Z_i - \mathbb{E}[Z_i]| > CS\epsilon/2 \ | \ \mathcal{B}^c).
        \end{equation}
    Now, we can use the result from Lemma \ref{azuma-finite-difference} to bound the right-hand side of Equation (\ref{total-prob-inequality-finite}) as:
    \begin{equation}
        \label{eq:azuma-substitution}
        \text{Pr}(|Z_i - \mathbb{E}[Z_i]| > CS\epsilon/2 \ | \ \mathcal{B}^c) + \text{Pr}(\mathcal{B}) \leq 2\exp\left(\frac{-C^2 S^2 \epsilon^2}{8 \sum_{i=1}^{SC} \alpha_i^2}\right) + d_3 S\exp(-\beta S^{\frac{2}{4i+1}}).
    \end{equation} 
    Since the slowest growing exponential in Equation (\ref{eq:azuma-substitution}) is $S^{\frac{2}{4i+1}}$, we can again lower bound the right-hand side to get the result: $$\text{Pr}(|Z_i - \mathbb{E}[Z_i]| > CS\epsilon/2 ) \leq 2 \exp(-\beta \epsilon^2 S^{\frac{1}{4i+1}}).$$
\end{proof}

\begin{lemma} (Adapted from~\cite{pedarsani2017phasecode})
    \label{tree-neighborhood}
    For a sufficiently large $S$, all neighborhoods  $\mathcal{N}_{e}^{2i}$ are trees.
\end{lemma}

\begin{proof}
    To prove that $\mathcal{N}_{e}^{2i}$ is a tree for all $i$, we start with a fixed $\mathcal{N}_e^{2i}$ that is a tree approaching $i^* > i$. Let $C_{i}$ be the number of check nodes and $V_i$ be the number of variable nodes in the tree at this point. We will start by showing that, with a high probability, the size of the tree neighborhood grows slower than $O(\log^i(S))$.
    Using the law of total probability, we can write:
    \begin{equation}
        \label{tree-probability-expansion}
        \text{Pr}(\mathcal{T}_i^c) \leq \text{Pr}(V_i > \kappa_1 \log^i(S)) + \text{Pr}(C_i > \kappa_2 \log^i(S)) + \text{Pr}(\mathcal{T}_i^c \ | \ Vi < \kappa_1 \log^i(S), C_i > \kappa_2 \log^i(S))).
    \end{equation}
    If we denote $\alpha_i = \text{Pr}(V_i > \kappa_1 \log^i(S))$, then we bound $\alpha_i$ as: 
    \begin{align}
    \alpha_i &\leq \alpha_{i-1} + \text{Pr}(V_i > \kappa_1 \log^i(S) \mid V_{i-1} < \kappa_1 \log^i(S)) \notag \\
             &\leq \alpha_{i-1} + \text{Pr}(V_i > \kappa_1 \log^i(S) \mid C_{i-1} < \kappa_2 \log^i(S)),
    \label{tree-alpha-expansion}
    \end{align}
    where the last inequality comes from the fact that for each variable node added to the tree, a maximum of $C-1$ check nodes are also added to the tree. To bound the last term in Equation (\ref{tree-alpha-expansion}), we would need to analyze the degree distribution of each check node. We will examine the number of check nodes at exactly depth $i$ in the graph $M_i$, where $M_i < C_i$. Let $D_i$ be the degree of each check node. Finding the probability that $V_i > \kappa_1 \log^i(S)$ is equivalent to finding the probability that $\sum_j D_j > \kappa_3 \log^i(S) $ since each edge of each $M_i$ check nodes is connected to a variable node in $V_i$. Therefore, we can bound the right-hand side of Equation (\ref{tree-alpha-expansion}) as: 
    \begin{equation}
        \label{poisson-tree-expansion}
        \text{Pr}(V_i > \kappa_1 \log^i(S) \ | \ C_{i-1} < \kappa_2 \log^i(S)) \leq \text{Pr}\left(\sum_{j=i}^{M_i}D_j \geq \kappa_3 \log^i(S)   \right).
    \end{equation}

    We know that check node degrees are Poisson distributed, but each group $c$ has a different mean $1/\eta_c$. For an accurate bound, we will assume all check node degrees are Poisson distributed with the highest mean from $\{\frac{1}{\eta_c}\}_{c \in [C]}$.  Let $D_{j, c_{max}}$ denote the degree of a check node from group $c_{\text{max}} \coloneqq \argmax_c (1/\eta_c)$. Then, we can write a final bound as:
    \begin{equation}
        \label{max-poisson-bound}
        \text{Pr}\left(\sum_{j=i}^{M_i}D_j \geq \kappa_3 \log^i(S)   \right) \leq \text{Pr}\left(\sum_{j=i}^{M_i}D_{j, c_{\text{max}}} \geq \kappa_3 \log^i(S)   \right) \leq \left(\frac{eM_i/\eta_{c_{\text{max}}}}{\kappa_3 \log^i (S)} \right)^{\kappa_3 \log^i (S)} \leq O(1/S),
    \end{equation}
    meaning that $\alpha_i \leq \alpha_{i-1} + O(1/S)$. Since we can bound the number of variable nodes in the graph with this high probability, we know that the same bound applies to the check nodes.
    Now that we have shown that the tree grows no larger than $O(\log^i(S))$ with probability $1/S$, we can move on to the probability of creating a tree. The rest of the proof follows closely from~\cite{richardson2001capacity}, other than the assumption of uniform check node degree.
    
    Let $C_{i, c}$ be the number of check nodes from group $c$ present in the graph. Let $t$ be the number of new edges to check nodes in $c$ from the variable nodes in $\mathcal{N}_e^{2i + 1}$ that do not create a loop. Then, the probability that the next revealed edge does not create a loop is: $$\frac{B_c - C_{i, c} - t}{B_c} \geq 1 - \frac{C_{i^*, c}}{B_c},$$ where $c$ is the group that the next edge connects to, and assuming that $i^*$ is sufficiently large. If $\mathcal{N}_e^{2i}$ is a tree, then the probability that $\mathcal{N}_e^{2i+1}$ is a tree depends on whether all the new edges from the variable nodes in $\mathcal{N}_e^{2i}$ connect to check nodes that were not already in the tree. We can lower bound the probability as the following: $$\text{Pr}(\mathcal{N}_e^{2i + 1} \text{ is tree-like }| \ \mathcal{N}_e^{2i} \text{ is tree-like }) \geq \min_c \left(1 - \frac{C_{i^*, c}}{B_c}\right)^{C_{i+1, c} - C_{i, c}}.$$ If we then assume that $\mathcal{N}_e^{2i+1}$ is a tree, the probability of whether $\mathcal{N}_e^{2i + 2}$ is a tree can be found similarly. Assuming that $k$ edges have been revealed going from check nodes to variable nodes at depth $2i +2$ without creating a loop, the probability of the next edge creating a loop is: $$\frac{(S - V_i - k)C}{SC - V_i - k} \geq \left(1 - \frac{V_i^*}{S}\right). $$ This similarly gives a lower bound on the probability: $$\text{Pr}(\mathcal{N}_e^{2i + 2} \text{ is tree-like }| \ \mathcal{N}_e^{2i + 1} \text{ is tree-like }) \geq  \left(1 - \frac{V_{i^*}}{S}\right)^{V_{i+1} - V_{i}}.$$ This means that the probability that $\mathcal{N}_e^{2i^*}$ is a tree is lower bounded by: $$\min_c \left(1 - \frac{C_{i^*, c}}{B_c}\right)^{C_{i^*, c}} \left(1 - \frac{V_{i^*}}{S}\right)^{V_{i^*}}.$$
    For sufficiently large $k$ and fixed $i^*$, we can upper bound the probability that $\mathcal{N}_e^{2i^*}$ is not a tree as:
    \begin{equation}
        \label{tree-square-bound}
        \text{Pr}(\mathcal{N}_e^{2i^*} \text{ is not tree-like}) \leq \max_{c} \frac{M_{i^*}^2 + \frac{1}{\eta_c}C_{i^*, c}^2}{S}.
    \end{equation}
    We previously showed that check node size at depth $i$ is upper bounded by $O(\log^i(S))$ with a high probability, which means that we can bound Equation (\ref{tree-square-bound}) again as: \begin{equation}
        \text{Pr}(\mathcal{N}_e^{2i^*} \text{ is not tree-like}) \leq O(\log^{i^*}(S)/S).
    \end{equation}
\end{proof}

\begin{definition}  (Adapted from~\cite{li2015spright})
    Consider a subset $\mathcal{S}$ of $\mathcal{G}(S, \{\eta_{c}\}_{c \in [C]}, C, \{{\mathbf{M}}_{c}\}_{c \in [C]})$ with $|\mathcal{S}| = k < \epsilon S$. If, for a neighborhood of check nodes in group $c$ $\mathcal{N}_c(\mathcal{S})$, the condition is met that $|\mathcal{N}_c(\mathcal{S})| > |\mathcal{S}|/2$, then $\mathcal{G}(S, \{\eta_{c}\}_{c \in [C]}, C, \{{\mathbf{M}}_{c}\}_{c \in [C]})$ is an $(\epsilon, 1/2, C)$-expander.
\end{definition}
\begin{lemma}
\label{expander probability}
    The probability that $\mathcal{G}(S, \{\eta_{c}\}_{c \in [C]}, C, \{{\mathbf{M}}_{c}\}_{c \in [C]})$ is an expander graph is at least $1 - O(1/S)$. 
\end{lemma}
\begin{proof}
    Given a subset of the variable nodes $\mathcal{S}$ where $|\mathcal{S}| = k < \epsilon S$, we can make a neighborhood of check nodes $\mathcal{N}_c(\mathcal{S})$. The graph $\mathcal{G}(S, \{\eta_{c}\}_{c \in [C]}, C, \{{\mathbf{M}}_{c}\}_{c \in [C]})$ is considered an expander graph if $|\mathcal{N}_c(\mathcal{S})| \geq k/2$ for at least 1 group $c$. We can write the probability of this as: 

 \begin{equation}
    \label{eq: expanderbasic}
    \text{Pr}(\mathcal{S} \text{ is not an expander}) \leq {S \choose k} \prod_{c=1}^C {B_c \choose k/2} \left(\frac{k}{2Bc}\right)^k, 
\end{equation}
where ${S \choose k} $ represents the total number of possible subsets that can be made, ${B_c \choose k/2}$ represents the number of possible check node subsets that can be chosen from $B_c$, and $\left(\frac{k}{2Bc}\right)^k$ represents the probability that the subset of check nodes contains all the neighbors of the vertices in $\mathcal{S}$. 
Then we can use the inequality that ${a \choose b} \leq (ae/b)^b$ to rewrite Equation (\ref{eq: expanderbasic}) as:
$$\text{Pr}(\mathcal{S} \text{ is not an expander}) \leq \left(\frac{Se}{k}\right)^k \left(\frac{ke}{2S}\right)^{Ck/2} \prod_{c=1}^C \left(\frac{1}{\eta_c}\right)^{k/2}.$$ This can again be bounded as:
\begin{equation}
   \text{Pr}(\mathcal{S} \text{ is not an expander}) \leq \left(\frac{k}{S}\right)^{k(\frac{C}{2}-1)}c^k, \quad c = e\left(\prod_{c=1}^C \left(\frac{e}{\eta_c}\right)\right)^{1/2}.
\end{equation} Given that $\frac{C}{2} -1 \geq \frac{1}{2}$, we can bound the probability again as: \begin{equation}
    \text{Pr}(\mathcal{S} \text{ is not an expander}) \leq \left(\frac{kc^2}{S}\right)^{k/2}.
\end{equation}  
Clearly, $\mathcal{S}$ being an expander is dependent on the number of vertices $k$. We can take 2 edge cases where $k = O(S)$ and $k = O(1)$: \begin{equation}
    \text{Pr}(\mathcal{S} \text{ is not an expander}) = \begin{cases}
    e^{-S\log(\frac{1}{\epsilon c^2})}, \quad k = \epsilon S \text{ with } \epsilon < \frac{1}{2c^2} \\
    O(1/S), \quad k = O(1).
\end{cases}
\end{equation}
For the very sparse regime $0 \leq \delta \leq 1/3$, we are unlikely to see case 1. Therefore, we can show that $\mathcal{S}$ is an $\epsilon$-expander graph with probability $1 - O(1/S).$ For the less sparse regime, defined in \cite{li2015spright} as $\delta = 1 - 1/C,$ we may have some overlap in the indices of each frequency index $\mathbf{k}, $ which means that the assumption of independence across groups no longer holds. Similar to the proof of Lemma 9 in \cite{li2015spright}, we show that for a given subset $\mathcal{S}$ of left nodes, we can still guarantee the same expander condition that allows peeling to continue by bounding the probability that there exists a subset of left nodes $\mathcal{S}$ such that none of them are singletons. Given an arbitrary left node $\mathbf{k}, $ its right neighbors are all multi-tons if and only if for each $c$, there is some $\mathbf{k}'$ such that $\mathbf{M}_c^T\mathbf{k} = \mathbf{M}_c^T\mathbf{k}'.$ Using the less sparse construction of $\mathbf{M}_c$ detailed in Section B.2 of \cite{li2015spright}, we can construct a set $\mathcal{S}$ such that for every pair $\mathbf{k} \neq \mathbf{k}' \in \mathcal{S}$ one has  $\mathbf{M}_c^T\mathbf{k} = \mathbf{M}_c^T\mathbf{k}'$ for all $c \neq c^*$, while $\mathbf{M}_{c^*}^T \mathbf{k} \neq \mathbf{M}_{c}^T \mathbf{k}'.$ Since there are at least 2 nodes for each group $c \in [C]$ needed to force a multi-ton, we know that the miniumum size of $\mathcal{S}$ is $k \geq 2^C$. Taking the worst case scenario where there are $k/2^{C-1}$ nodes in one group that are in $\mathcal{S}$ (with $2$ nodes connected to the other $C-1$ groups), the total number of possible left nodes $\mathbf{k}$ that allow for this collision is $N/B_c.$ The probability that for an arbitrary set of $k/2^{C-1}$ left nodes to all hash to the same right node in all subsampling groups is found as: \begin{equation}
    \text{Pr}(\mathcal{S} \text{ is not an expander}) \leq {S \choose k}\prod_{c=1}^C {N/B_c \choose k/2^{C-1}} \left( \frac{B_c k}{2^{C-1} N} \right)^k.
\end{equation}
We can once again use the inequality ${a \choose b} \leq (ae/b)^b$ to bound this as: \begin{equation}
     \text{Pr}(\mathcal{S} \text{ is not an expander}) \leq \left(\frac{Se}{k} \right)^k \prod_{c=1}^C \left(\frac{S\eta_ck}{2^{C-1}N} \right)^{k \left(1- \frac{1}{2^{C-1}} \right)}e^{k/2^{C-1}}. 
\end{equation}
Substituting $k = 2^C$ and $S = N^{1 - 1/C}$ gives us the final bound: $$\text{Pr}(\mathcal{S} \text{ is not an expander}) \leq O(S^{-\frac{2^c - 2C}{C-1}}),$$ which is upper bounded by $O(1/S)$ for $C \geq 3.$ 
\end{proof}
If we have an expander graph, the average degree within at least one group $c$ is less than 2, meaning at least 1 group $c$ has a singleton that can be peeled. This allows for another iteration of peeling. The number of samples required for the algorithm scales with $O(P\sum_{c=1}^CB_c) = O(Sn)$. The computational complexity is $O(PB\log B) = O(Sn\log S) = O(Sn\log N)$ since $S = O(N^\delta)$.

\subsection{Proof of Proposition \ref{proposition1}}
\label{proposition1-proof}
\begin{proof}
Given a singleton pair $(\mathbf{k}, F[\mathbf{k}])$ in group $c$, the angle quantization is written as: 
\begin{equation}
    \text{arg}[U_{p}] = \text{arg}[F[\mathbf{k}]] + \left(\sum_{i=1}^n \frac{2\pi}{q_i} (\mathbf{d}_{p})_i k_i \right) + Y_p,
\end{equation}
where $Y_p$ is a uniform random variable over $[-\pi, \pi)$.
Similarly, for $U_{p, r}$, we have: \begin{equation}
     \text{arg}[U_{p, r}] = \text{arg}[F[\mathbf{k}]] + \left(\sum_{i=1}^n \frac{2\pi}{q_i} (\mathbf{d}_{p, r})_i k_i \right) + Y_{p, r}, 
\end{equation}
where $\frac{2\pi}{q_i} (\mathbf{d}_{p, r})_i k_i = \frac{2\pi}{q_i} (\mathbf{d}_{p})_i k_i$ for all $i \neq r$ and $Y_{p, r}$ is another uniform random variable over $[-\pi, \pi)$.
The quantization of the ratio of $U_{p,r}$ and $U_{p}$ is written as, 
\begin{equation}
    \label{argq-ratio}
    \text{arg}_{}[U_{p, r}/U_p] =\frac{2\pi}{q_r}k_r + Y_{p, r} - Y_p, 
\end{equation}
which satisfies:
$$\text{Pr}(|Y_p| \geq \alpha) \leq \text{Pr}(|W_p| \geq |F[\mathbf{k}]| \sin(\alpha))$$ for $0 \leq \alpha \leq \pi/2$. Using the definition of a Rayleigh random variable, the right-hand side is bounded by: \begin{equation}
    \label{rayleigh-bound}
    \text{Pr}(|W_p| \geq F[\mathbf{k}] \sin(\alpha)) \leq \exp\left(-\frac{|F[\mathbf{k}]|^2 \sin^2(\alpha)}{2\sigma^2/B}\right) = \exp\left(-\frac{\eta \sin^2(\alpha)\text{SNR}}{2}\right).
\end{equation}
 The final result is obtained as: \begin{equation}
    \begin{split}
            \text{Pr}(|Y_{p, r} - Y_p| \geq \pi/2q_r) \leq \text{Pr}(|Y_{p,r}| \geq \pi/2q_r) + \text{Pr}(|Y_{p}| \geq \pi/2q_r) \leq \\ 2 \exp\left(-\frac{\eta\sin^2(\pi/2q_r)\text{SNR}}{2}\right).
    \end{split}
\end{equation}
\end{proof}
\subsection{Proof of Theorem \ref{noisy}}
\label{noisy-full-proof}
Throughout this proof, we assume that each $q_i$ is a fixed constant (i.e., $q_i = O(1) \ \forall i \in 1 \dots n$). This lets us assume that $\log N = O(n)$, which will be necessary throughout this proof. We also drop the $c$ notation as all analysis is within a certain group, and is the same regardless of group.
\begin{proposition}
    Let $\mathbb{P}_F$ be the probability of $\mathcal{E},$ the event that the robust bin detection algorithm makes an error in $CS = O(S)$ iterations (i.e. $F \neq \hat{F}$) within some group $c$. The probability of error can be bounded as: $$\mathbb{P}_F \leq O(1/S).$$
\end{proposition}
\begin{proof}
    Using the law of total probability, we rewrite the probability that NR-GFast fails as: \begin{equation*}
        \label{failure-total-probably-nr}
        \text{Pr}(\hat{F} \neq F) = \text{Pr}(\hat{F} \neq F \ | \ \mathcal{E})\text{Pr}(\mathcal{E}) + \text{Pr}(\hat{F} \neq F \ | \ \mathcal{E}^c)\text{Pr}(\mathcal{E}^c),
    \end{equation*}
    which can again be bounded as \begin{equation*}
        \text{Pr}(\hat{F} \neq F) \leq  \text{Pr}(\mathcal{E}) + \text{Pr}(\hat{F} \neq F \ | \ \mathcal{E}^c).
    \end{equation*}
    We can use the result of Theorem \ref{noiseless} showing that $\text{Pr}(\hat{F} \neq F \ | \ \mathcal{E}^c) = O(1/S)$. Next, we must show that $\text{Pr}(\mathcal{E}) \leq O(1/S)$. Let $\mathcal{E}_b$ be the event of an error happening for a specific bin. Then, we can union bound over all bins in $c$ and all iterations to get a bound on $\text{Pr}(\mathcal{E})$: \begin{equation}
        \label{union-bound-error-bin-detection}
        \text{Pr}(\mathcal{E}) \leq \sum_{c=1}^CS^2  \eta_c \text{Pr}(\mathcal{E}_b).
    \end{equation}
    Therefore, to show that $\text{Pr}(\mathcal{E}) \leq O(1/S)$, we must show that $\text{Pr}(\mathcal{E}_b) \leq O(1/S^3)$.
    Using the law of total probability, we rewrite $\text{Pr}(\mathcal{E}_b)$ as: \begin{equation}
    \label{eq: binning-probability-breakup}
    \text{Pr}(\mathcal{E}_b) \leq \sum_{\mathcal{F} \in \{\mathcal{H}_m, \mathcal{H}_Z\}} \text{Pr}(\mathcal{F} \leftarrow \mathcal{H}_S (\mathbf{k}, F[\mathbf{k}])) \text{ } + \sum_{\mathcal{F} \in \{\mathcal{H}_m, \mathcal{H}_Z\}} \text{Pr}(\mathcal{H}_S (\mathbf{k}, F[\mathbf{k}]) \leftarrow \mathcal{F}) + \text{ Pr}((\hat{\mathbf{k}}, \hat{F}[\hat{\mathbf{k}}]) \leftarrow (\mathbf{k}, F[\mathbf{k}])).
    \end{equation}
    Here, the following events are represented: 
    \begin{enumerate}
        \item The event $\mathcal{F} \leftarrow \mathcal{H}_S (\mathbf{k}, F[\mathbf{k}])$ represents a singleton bin being incorrectly classified as a zero-ton or multi-ton. 
        \item The event $\mathcal{H}_S (\mathbf{k}, F[\mathbf{k}]) \leftarrow \mathcal{F}$ represents a zero-ton or multi-ton being incorrectly classified as a singleton.
        \item The event $(\hat{\mathbf{k}}, \hat{F}[\hat{\mathbf{k}}]) \leftarrow (\mathbf{k}, F[\mathbf{k}])$ represents a singleton bin being classified with an incorrect index and value pair.
    \end{enumerate}
    It is shown in Propositions \ref{somethingelse-as-singleton}, \ref{singleton-as-somethingelse}, and \ref{wrong-singleton} that the probabilities of each of these events decrease exponentially with respect to $P_1$, so setting $P_1 = O(n) = O(\log N)$ lets us bound $\text{Pr}(\mathcal{E}_b) \leq O(1/N^3)$, and since $S$ is sublinear in $N$, this can be used to get the original result $\text{Pr}(\mathcal{E}) \leq S^2 \sum_{c=1}^C \eta_c \text{Pr}(\mathcal{E}_b) \leq O(1/S)$. We note that the sample complexity is $O(P \sum_{c=1}^C B_c)= O(P_1nS) = O(Sn^2)$, and the computational complexity is $O(P\sum_{c=1}^CB_c\log B_c) = O(P_1nS \log S) = O(Sn^2\log N)$.
\end{proof}


\begin{proposition}
    \label{somethingelse-as-singleton}
    The probability of the event $\mathcal{H}_S (\mathbf{k}, F[\mathbf{k}]) \leftarrow \mathcal{H}_Z$ is upper bounded by $\exp\left(-\frac{P_1}{2}(\sqrt{1 + 2\gamma} - 1)^2\right)$, and the probability of the event $\mathcal{H}_S (\mathbf{k}, F[\mathbf{k}]) \leftarrow \mathcal{H}_M$ is upper bounded by $4N^2\exp\left(-\epsilon \left(1 - \frac{2 \gamma \nu^2}{L\rho^2}\right)^2 P_1 \right)$.
\end{proposition}
\begin{proof}
    Let us start with the event that a zero-ton is incorrectly identified as a singleton: \begin{equation}
        \text{Pr}(\mathcal{H}_S(\hat{\mathbf{k}}, \hat{F}[\hat{\mathbf{k}}]) \leftarrow \mathcal{H}_Z) \leq \text{Pr}\left(\frac{1}{P} ||\mathbf{W}||^2 \geq  (1 + \gamma)\nu^2\right).
    \end{equation}
    By applying Lemma \ref{lemma:qsft-tail-bound}, we bound the probability of this event as: \begin{equation}
        \text{Pr}(\mathcal{H}_S(\hat{\mathbf{k}}, \hat{F}[\hat{\mathbf{k}}]) \leftarrow \mathcal{H}_Z) \leq \text{Pr}\left(\frac{1}{P_1} ||\mathbf{W}||^2 \geq  (1 + \gamma)\nu^2\right) \leq \exp\left(-\frac{P_1}{2}(\sqrt{1 + 2\gamma} - 1)^2\right),
    \end{equation}
    where we let $\mathbf{g} = \mathbf{0}$ and $\mathbf{v} = \mathbf{W}$. Next, we will examine the event that a multi-ton is incorrectly identified as a singleton. \begin{equation}
        \text{Pr}(\mathcal{H}_S(\hat{\mathbf{k}}, \hat{F}[\hat{\mathbf{k}}]) \leftarrow \mathcal{H}_M) \leq \text{Pr}\left(\frac{1}{P_1} ||\mathbf{g} + \mathbf{W}||^2 \leq  (1 + \gamma)\nu^2\right),
    \end{equation}
    where $\mathbf{g} = \mathbf{S}(\boldsymbol{\alpha} - \hat{F}[\hat{\mathbf{k}}]\mathbf{e}_{\hat{\mathbf{k}}})$. This can again be upper bounded using the law of total probability: \begin{equation}
        \label{multiton-as-singleton}
        \text{Pr}(\mathcal{H}_S(\hat{\mathbf{k}}, \hat{F}[\hat{\mathbf{k}}]) \leftarrow \mathcal{H}_M) \leq \text{Pr}\left(\frac{1}{P_1} ||\mathbf{g} + \mathbf{W}||^2 \leq  (1 + \gamma)\nu^2 \,\big|\, \frac{\|\mathbf{g}\|^2}{P_1} \geq 2\gamma\nu^2  \right) + \text{Pr} \left(\frac{\|\mathbf{g}\|^2}{P_1} \leq 2\gamma\nu^2 \right).
    \end{equation}
    From Lemma \ref{lemma:qsft-tail-bound}, we know that the left term can be bounded as $\exp(-\frac{P_1 \gamma^2}{2 + 8\gamma})$. For the second term, let $\boldsymbol{\beta} = \boldsymbol{\alpha} - \hat{F}[\hat{\mathbf{k}]}\mathbf{e}_{\hat{\mathbf{k}}}$, and label $\mathbf{g} = \mathbf{S}\boldsymbol{\beta}$. Defining $\mathcal{L} \coloneqq \text{supp}(\boldsymbol{\beta})$, let $\mathbf{S}_{\mathcal{L}}$ refer to the corresponding columns of $\mathbf{S}$ and $\boldsymbol{\beta}_{\mathcal{L}}$ refer to the corresponding elements of $\boldsymbol{\beta}$ such that $\mathbf{S}\boldsymbol{\beta} = \mathbf{S}_{\mathcal{L}} \boldsymbol{\beta}_{\mathcal{L}}$. We consider two scenarios: 
    \begin{itemize}
        \item Case 1: $L = | \mathcal{L}| = O(1)$, i.e., the number of multi-tons is some arbitrary constant. We can bound the magnitude of $\| \mathbf
        S \boldsymbol{\beta} \|^2$ as: $$\lambda_{min} (\mathbf{S}_{\mathcal{L}}^H \mathbf{S}_{\mathcal{L}}) \| \boldsymbol{\beta}_{\mathcal{L}} \|^2 \leq \| \mathbf
        S_{\mathcal{L}} \boldsymbol{\beta}_{\mathcal{L}} \|^2 \leq \lambda_{max} (\mathbf{S}_{\mathcal{L}}^H \mathbf{S}_{\mathcal{L}}) \| \boldsymbol{\beta}_{\mathcal{L}} \|^2,$$
        which lets us bound $\text{Pr} \left(\frac{\|\mathbf{g}\|^2}{P_1} \leq 2\gamma\nu^2 \right) \leq \text{Pr}(\lambda_{min} (\frac{1}{P_1}\mathbf{S}_{\mathcal{L}}^H \mathbf{S}_{\mathcal{L}}) \leq 2\gamma\nu^2/(L\rho^2))$.
        The matrix $\frac{1}{P_1}\mathbf{S}_{\mathcal{L}}^H \mathbf{S}_{\mathcal{L}} \in \mathbb{C}^{L \times L}$ is a Hermitian matrix, which lets us bound $\lambda_{min} (\frac{1}{P_1}\mathbf{S}_{\mathcal{L}}^H \mathbf{S}_{\mathcal{L}})$ using the Gershgorin circle theorem: $$\lambda_{min} (\frac{1}{P_1}\mathbf{S}_{\mathcal{L}}^H \mathbf{S}_{\mathcal{L}}) \geq 1-L\mu,$$ where $\mu$ is defined as seen in Lemma \ref{lemma: mutual-coherence}.
        We can then further bound the probability: \begin{equation}
            \text{Pr} \left(\frac{\|\mathbf{g}\|^2}{P_1} \leq 2\gamma\nu^2 \right)  \leq \text{Pr}(1-L\mu \leq 2\gamma\nu^2/(L\rho^2)) = \text{Pr}\left(\mu \geq \frac{1}{L} \left(1 - \frac{2 \gamma \nu^2}{L\rho^2} \right)\right).
        \end{equation}
        Using the results of Lemma \ref{lemma: mutual-coherence} and setting $\mu_0 = \frac{2 \gamma \nu^2}{L\rho^2}$, we get the final bound: 
        \begin{equation}
        \label{prop3eqn}
             \text{Pr} \left(\frac{\|\mathbf{g}\|^2}{P_1} \leq 2\gamma\nu^2 \right) \leq 4N^2 \exp\left(-\frac{P_1}{8L^2} \left(1 - \frac{2 \gamma \nu^2}{L\rho^2}\right)^2 \right),
        \end{equation} where $\gamma < \frac{L\rho^2}{2\nu^2}$.
        \item Case 2: $L = |\mathcal{L}| = \omega(1)$. In this case, both $\boldsymbol{\beta}$ and $\mathbf{S}$ are made of fully random entries, meaning that the random vector $\mathbf{S}\boldsymbol{\beta}$ becomes Gaussian due to the central limit theorem with mean $0$ and covariance $\mathbb{E}[\mathbf{gg}^H] = L\rho^2 \mathbf{I}.$ This lets us apply Lemma \ref{lemma:qsft-tail-bound} to $\text{Pr} \left(\frac{\|\mathbf{g}\|^2}{P_1} \leq 2\gamma\nu^2 \right)$: 
        \begin{equation}
            \text{Pr} \left(\frac{\|\mathbf{g}\|^2}{P_1} \leq 2\gamma\nu^2 \right) \leq \exp\left(-\frac{P_1}{2} \left(1 - \frac{2 \gamma \nu^2}{L\rho^2}\right) \right),
        \end{equation}
        where $\gamma < \frac{L\rho^2}{2\nu^2}$. \\\\\ 
        Combining the bounds of Case 1 and 2, it is possible to find an absolute constant $\epsilon > 0$ that satisfy: \begin{equation}
            \label{final-multiton-bound}
            \text{Pr} \left(\frac{\|\mathbf{g}\|^2}{P_1} \leq 2\gamma\nu^2 \right) \leq 4N^2 \exp\left(-\epsilon\frac{P_1}{2} \left(1 - \frac{2 \gamma \nu^2}{L \rho^2}\right)^2\right),
        \end{equation}
        for $\gamma < \frac{\rho^2}{2\nu^2}$, where Equation (\ref{final-multiton-bound}) follows from~\cite{erginbas2023efficiently}.
    \end{itemize}
\end{proof}

\begin{lemma}
    \label{lemma: mutual-coherence}
    The mutual coherence of $\mathbf{S}$ is denoted as $\mu \coloneqq \max_{\mathbf{m} \neq \mathbf{k}} \frac{1}{P_1} |\mathbf{s}_{\mathbf{k}}^H \mathbf{s}_{\mathbf{m}}|$. For some $\mu_0 > 0$, the following inequality holds: \begin{equation}
       \text{Pr}(\mu > \mu_0) \leq 4N^2\exp\left(-\frac{\mu_0^2}{8}P_1\right),
    \end{equation}
    where $N$ is the number of possible vectors.
\end{lemma}
\begin{proof}
    The $p$-th entry of $\mathbf{s}_{\mathbf{k}}$ is denoted as $\prod_{i=1}^n \omega_{q_i}^{(\mathbf{d}_{p})_ik_i}$. If we let $\mathbf{y}$ denote $\mathbf{s}_{\mathbf{m}}^H \mathbf{s}_{\mathbf{k}}$, then $$\mathbf{y}_p = \prod_{i=1}^n \omega_{q_i}^{(\mathbf{d}_{p})_i (k_i - m_i)}.$$
    Since $(\mathbf{k} - \mathbf{m}) \mod \mathbf{q} \neq \mathbf{0}$ for all $\mathbf{k} \neq \mathbf{m}$, we know that each entry of $\mathbf{y}$ is a term from a uniform distribution $\{ \prod_{i=1}^n \omega_{q_i}^{(\mathbf{d}_{p})_i (k_i - m_i)} : (\mathbf{d}_{p})_i \in \mathbb{Z}_{q_i}, i = 1,\ldots,n \}$. From Lemma \ref{lemma: zero-vector-innerprod}, it is clear that this uniform selection implies that $\mathbb{E}[\mathbf{y}] = \mathbf{0}_{P_1}$. With this knowledge, we can apply the Hoeffding bound:
    \begin{equation}
        \text{Pr} \left(\frac{|\mathbf{s}_{\mathbf{m}}^H\mathbf{s}_{\mathbf{k}}|}{P_1} \geq \mu_0 \right) = \text{Pr}(\mathbf
        |{\mathbf{y}}| \geq P_1 \mu_0) \leq 2\text{Pr}(\text{Re}(\mathbf{y}) \geq P_1 \mu_0/\sqrt{2}) \leq 4\exp\left(-\frac{\mu_0^2}{8}P_1\right).
    \end{equation}
    The result is then be obtained by union bounding over all combinations of $\mathbf{m}$ and $\mathbf{k}$.
\end{proof} 
\begin{lemma}
    \label{lemma:qsft-tail-bound}
    (Restated from Lemma 2 of~\cite{erginbas2023efficiently}) Given $\mathbf{g} \in \mathbb{C}^P$ and a complex Gaussian vector $\mathbf{v} \sim \mathcal{CN}(0, \nu^2 \mathbf{I})$, the following tail bounds hold: $$\text{Pr}\left(\frac{1}{P}\|\mathbf{g} + \mathbf{v}\|^2 \geq \tau_1 \right) \leq \exp\left(-\frac{P}{2} \left(\sqrt{2\tau_1/\nu^2 - 1} - \sqrt{1+2\theta_0} \right)^2\right),$$
    $$\text{Pr}\left(\frac{1}{P}\|\mathbf{g} + \mathbf{v}\|^2 \leq \tau_2 \right) 
\leq \exp\left(-\frac{P}{2} \frac{\left(1 + \theta_0 - \tau_2/\nu^2 \right)^2}{1 + 2\theta_0}\right),$$
    for any $\tau_1, \tau_2$ that satisfy $$\tau_1 \geq \nu^2(1+\theta_0)^2 \geq \tau_2,$$ where $$\theta_0 = \frac{\|\mathbf{g}\|^2}{P\nu^2}.$$
\end{lemma}

\begin{proposition} (Adapted from Proposition 5 of~\cite{erginbas2023efficiently})
    \label{singleton-as-somethingelse}
    The probability of the event $\mathcal{H}_Z \leftarrow \mathcal{H}_S (\mathbf{k}, F[\mathbf{k}])$ is upper bounded by $\exp\left(-\frac{P_1}{2}\left(\frac{(\rho^2/\nu^2 - \gamma)^2}{1 + 2\rho^2/\nu^2}\right)\right)$, and the probability of the event $\mathcal{H}_M \leftarrow \mathcal{H}_S (\mathbf{k}, F[\mathbf{k}])$ is upper bounded by $\exp\left(-\frac{P_1}{2}(\sqrt{1 + 2\gamma} - 1)^2\right) + 2 \exp \left(-\frac{P_1 \rho^2 \sin^2(\pi/\kappa)}{2 \nu^2}\right) + 2n\exp\left(-\frac{1}{2q_{\text{max}}}\epsilon^2 P_1 \right)$, where $q_{\text{max}} \coloneqq \max \mathbf{q}$.
\end{proposition}
\begin{proof}
    We will begin with the case where a singleton is mistakenly identified as a zero-ton:
    \begin{equation}
        \text{Pr}(\mathcal{H}_Z \leftarrow \mathcal{H}_S(\mathbf{k}, F[\mathbf{k}]) \leq \text{Pr}\left(\frac{1}{P_1} \| F[\mathbf{k}]\mathbf{s}_{\mathbf{k}} + \mathbf{W}\|^2 \leq (1 + \gamma)\nu^2 \right).
    \end{equation}
    Using Lemma \ref{lemma:qsft-tail-bound} with $\mathbf{g} = F[\mathbf{k}]\mathbf{s}_{\mathbf{k}}$ and $\mathbf{v} = \mathbf{W}$, we get the final bound: \begin{equation}
        \label{eq: singleton-as-zeroton-final}
        \text{Pr}(\mathcal{H}_Z \leftarrow \mathcal{H}_S(\mathbf{k}, F[\mathbf{k}]) \leq \text{Pr}\left(\frac{1}{P_1} \| F[\mathbf{k}]\mathbf{s}_{\mathbf{k}} + \mathbf{W}\|^2 \leq (1 + \gamma)\nu^2 \right) \leq \exp\left(-\frac{P_1}{2}\left(\frac{(\rho^2/\nu^2 - \gamma)^2}{1 + 2\rho^2/\nu^2}\right)\right),
    \end{equation}
    which holds for $\gamma < \rho^2/\nu^2$.
    Next we will examine the case where a singleton is mistakenly identified as a multi-ton, where we want to find the probability of the event: $$\text{Pr}(\mathcal{H}_M \leftarrow \mathcal{H}_S(\mathbf{k}, F[\mathbf{k}]) = \text{Pr} \left( \frac{1}{P} \|\mathbf{U} -  \hat{F}[\hat{\mathbf{k}}]\mathbf{s}_{\hat{\mathbf{k}}}\|^2 \geq (1 + \gamma)\nu^2 \right),$$ which we will denote as $\mathcal{E}$. Then, using the law of total probability, we write: \begin{equation}
        \text{Pr}(\mathcal{E}) \leq \text{Pr}\left(\mathcal{E} \ | \ \hat{\mathbf{k}} = \mathbf{k} \ \bigcap \ \hat{F}[\hat{\mathbf{k}}] = F[\mathbf{k}] \right) + \text{Pr} \left(\hat{\mathbf{k}} \neq \mathbf{k} \ \bigcup \ \hat{F}[\hat{\mathbf{k}}] \neq F[\mathbf{k}] \right)
        \end{equation}
    The first term is the same as the probability that a zero-ton is misclassified as a singleton, which we know from Proposition \ref{somethingelse-as-singleton} can be upper bounded by $\exp\left(-\frac{P_1}{2}(\sqrt{1 + 2\gamma} - 1)^2\right)$. The second term, using a union bound, can again be bounded as: $$ \text{Pr} \left(\hat{\mathbf{k}} \neq \mathbf{k} \ \bigcup \ \hat{F}[\hat{\mathbf{k}}] \neq F[\mathbf{k}] \right)
        \leq \text{Pr}(\hat{\mathbf{k}} \neq \mathbf{k}) + \text{Pr}(\hat{F}[\hat{\mathbf{k}}] \neq F[\mathbf{k}]) \leq 2\text{Pr}(\hat{\mathbf{k}} \neq \mathbf{k})+ \text{Pr}(\hat{F}[\hat{\mathbf{k}}] \neq F[\mathbf{k}] \ | \ \hat{\mathbf{k}} = \mathbf{k}).$$ The first term is given by Lemma \ref{singleton-search}, while the second term is the error probability of a $\kappa$-point PSK signal, with symbols from $supp(F) =\mathcal{X} \coloneqq \{\rho, \rho \phi, \rho \phi^2 , \cdots , \rho \phi^{\kappa-1}\}$. The symbol error probability of this PSK signal can be used to bound the second term: \begin{equation}
        \text{Pr}(\hat{F}[\hat{\mathbf{k}}] \neq F[\mathbf{k}] \ | \ \hat{\mathbf{k}} \neq \mathbf{k}) \leq 2 \exp \left(-\frac{P_1 \rho^2 \sin^2(\pi/\kappa)}{2 \nu^2}\right).
    \end{equation}
\end{proof}

\begin{proposition}
    \label{wrong-singleton}
    The probability of the event  $(\hat{\mathbf{k}}, \hat{F}[\hat{\mathbf{k}}]) \leftarrow (\mathbf{k}, F[\mathbf{k}])$ is upper bounded by: $$4N^2\exp\left(-\frac{P_1}{32} \left(1 - \frac{ \gamma \nu^2}{\rho^2}\right)^2 \right) + \exp\left(-\frac{P_1 \gamma^2}{2 + 8\gamma} \right).$$
\end{proposition} This event occurs when a singleton $(\hat{\mathbf{k}}, \hat{F}[\hat{\mathbf{k}}])$ different from $(\mathbf{k}, F[\mathbf{k}])$ passes the singleton verification test, which occurs when the following inequality is satisfied: $$\frac{1}{P} || \mathbf{U} - \hat{F}[\hat{\mathbf{k}}] \mathbf{s}_{\hat{\mathbf{k}}} ||^2 \leq (1 + \gamma)\nu^2.$$ Using the law of total probability, this inequality can be rewritten as:
\begin{equation}
    \begin{split}
        \text{Pr}((\hat{\mathbf{k}}, \hat{F}[\hat{\mathbf{k}}]) \leftarrow (\mathbf{k}, F[\mathbf{k}])) &= \text{Pr}\left(\frac{1}{P_1}\Big\|F[\mathbf{k}]\mathbf{s}_{\mathbf{k}} - \hat{F}[\hat{\mathbf{k}}] \mathbf{s}_{\hat{\mathbf{k}}} + \mathbf{W}\Big\|^2 \leq (1 + \gamma)\nu^2 \right) \\
        &\leq \text{Pr}\left(\frac{1}{P_1}\Big\|\mathbf{S}\boldsymbol{\beta} + \mathbf{W}\Big\|^2 \leq (1 + \gamma)\nu^2 \,\Big|\, \frac{1}{P_1}\Big\|\mathbf{S}\boldsymbol{\beta} \Big\|^2 \geq 2\nu^2\gamma \right) \\
        &\quad + \text{Pr}\left (\frac{1}{P_1}\Big\| \mathbf{S}\boldsymbol{\beta} \Big\|^2 \leq 2\nu^2\gamma \right).
    \end{split}
\end{equation}
From Lemma \ref{lemma:qsft-tail-bound}, the first term can be upper bounded as $\exp(-\frac{P_1 \gamma^2}{2 + 8\gamma})$. $\boldsymbol{\beta}$ has 2 non-zero terms selected from $\mathcal{X}$. This is because $F[\mathbf{k}]$ and $\hat{F}[\hat{\mathbf{k}}]$ are non-zero and do not cancel each other out when computing $F[\mathbf{k}]\mathbf{s}_{\mathbf{k}} - \hat{F}[\hat{\mathbf{k}}] \mathbf{s}_{\hat{\mathbf{k}}}$. Since the size of $\text{supp}(\boldsymbol{\beta})$ is $L = 2$, we can apply the same bounds as applied in Case 1 of Proposition \ref{somethingelse-as-singleton}, Equation (\ref{prop3eqn}): \begin{equation}
    \text{Pr}\left (\frac{1}{P_1}\Big\| \mathbf{S}\boldsymbol{\beta} \Big\|^2 \leq 2\nu^2\gamma \right) \leq 4N^2\exp\left(-\frac{P_1}{32} \left(1 - \frac{ \gamma \nu^2}{\rho^2}\right)^2 \right)
\end{equation}
for $\gamma < \rho^2/\nu^2$.
\begin{lemma}
    \label{singleton-search}
    The singleton search error probability of NR-GFast is upper bounded as: \begin{equation}
        \text{Pr}(\hat{\mathbf{k}} \neq  \mathbf{k}) \leq 2n\exp\left(-\frac{1}{2q_{\text{max}}}\epsilon^2 P_1 \right),
    \end{equation}
    where $q_{\text{max}} \coloneqq \max \mathbf{q}$, and for some constant $\epsilon > 0$.
\end{lemma}
\begin{proof}
    Referring to Proposition \ref{proposition1}, we know that $\text{Pr}(Z_{p, r} = 0) = p_0$, where $p_0 - p_i \geq \epsilon$ for all $i \neq 0$. For some $r$, the probability of the incorrect $\hat{k}_r$ can be bounded with: \begin{equation}
    \label{singleton-search-error-firstbound}
    \begin{split}
        \text{Pr}(\hat{k}_r \neq k_r) = \text{Pr}\left(\argmax_{a \in \mathbb{Z}_{q_r}} \sum_{p \in [P_1]}\mathbbm{1}\{a = \text{arg}_{q_r}[U_{p,r}/U_p]\} \neq \hat{k}_r\right)\\ = \text{Pr}\left(\argmax_{a \in \mathbb{Z}_{q_r}} \sum_{p \in [P_1]}\mathbbm{1}\{a = Z_{p,r}\} \neq 0\right). 
    \end{split}
    \end{equation} 
    Let $\mathbf{p} \coloneqq \{p_0, \cdots , p_{q_r - 1}\}$, and let $\hat{\mathbf{p}} = \frac{1}{n} \text{Multinomial}(P_1, \mathbf{p})$. Then, we bound Equation (\ref{singleton-search-error-firstbound}) with: \begin{equation}
    \begin{split}
        \text{Pr}\left(\argmax_{a \in \mathbb{Z}_{q_r}} \sum_{p \in [P_1]}\mathbbm{1}\{a = Z_{p,r}\} \neq 0\right)  \leq \text{Pr}(\|\mathbf{p} - \hat{\mathbf{p}}\| \geq \epsilon) \\ \leq 2\exp\left(-\frac{1}{2q_r} \epsilon^2 P_1 \right),
    \end{split}
    \end{equation}
    where the last line comes from Lemma 5 of~\cite{erginbas2023efficiently}. Union bounding over all $n$ bits in $\mathbf{k}$, we get the final result: \begin{equation}
        \label{final-singleton-search-bound}
        \text{Pr}(\hat{\mathbf{k}} \neq \mathbf{k}) \leq \sum_{r=1}^n 2\exp\left(-\frac{1}{2q_r} \epsilon^2 P_1 \right) \leq 2n\exp\left(-\frac{1}{2q_{\text{max}}} \epsilon^2 P_1 \right).
    \end{equation}
\end{proof}
\end{document}